\documentclass[11pt,a4paper]{article}
\usepackage{amsmath, amsfonts, amssymb}
\usepackage{a4wide}
\usepackage{parskip}
\usepackage{enumitem}
\usepackage{xcolor}

\usepackage{hyperref}
\usepackage{booktabs}
\usepackage{caption}
\usepackage{siunitx}
\usepackage{tabularx}
\usepackage{comment}
\usepackage{multirow}
\usepackage{adjustbox}
\usepackage{mathtools}
\usepackage{caption}
\usepackage{parskip}
\usepackage{graphicx}
\usepackage{adjustbox}
\usepackage{multirow}
\usepackage{amsthm}
\usepackage{bm, makecell}
\usepackage{lscape}
\usepackage{rotating}
\usepackage{floatrow}
\usepackage[noadjust]{cite}

\newtheorem{theorem}{Theorem}[section]
\newtheorem{lemma}[theorem]{Lemma}
\newtheorem{corollary}[theorem]{Corollary}
\newtheorem{proposition}[theorem]{Proposition}
\theoremstyle{definition}

\newtheorem{remark}[theorem]{Remark}

\numberwithin{equation}{section}

\newcommand{\Min}{\textnormal{min}}
\newcommand{\Tor}{\textnormal{Tor}}

\usepackage{tikz,xcolor,hyperref}
\usepackage{mathdots}

\definecolor{lime}{HTML}{A6CE39}
\DeclareRobustCommand{\orcidicon}{%
	\begin{tikzpicture}
		\draw[lime, fill=lime] (0,0) 
		circle [radius=0.16] 
		node[white] {{\fontfamily{qag}\selectfont \tiny ID}};
		\draw[white, fill=white] (-0.0625,0.095) 
		circle [radius=0.007];
	\end{tikzpicture}
	\hspace{-2mm}
}

\foreach \x in {A, ..., Z}{%
	\expandafter\xdef\csname orcid\x\endcsname{\noexpand\href{https://orcid.org/\csname orcidauthor\x\endcsname}{\noexpand\orcidicon}}
}

\begin{document}
	\date{}
		\title{Constacyclic codes of length $np^s$ over $\frac{\mathbb{F}_{p^m}[u]}{\langle u^t\rangle}$: Torsions and Cardinalities}
		\author{{\bf Akanksha Tiwari  \footnote{email: {\tt akankshafzd8@gmail.com} }\orcidA{}, \bf Pramod Kanwar\footnote{email: {\tt kanwar@ohio.edu}},  and \bf Ritumoni Sarma\footnote{email: {\tt ritumoni407@gmail.com}}\orcidC{}} \\\\ $*,\ddagger$ Department of Mathematics \\ Indian Institute of Technology Delhi\\Hauz Khas, New Delhi-110016, India\\\\ $\dagger$ Department of Mathematics \\ Ohio University-Zanesville \\Zanesville, Ohio, U.S.A.}
  
\maketitle

\begin{abstract}
The purpose of this article is to study constacyclic codes of length $np^s$ over $R^t:=\frac{\mathbb{F}_{p^m}[u]}{\langle u^t \rangle },$ where $t$ is a natural number and $\gcd(n,p)=1$. We give generators of all the ideals of  $R^{t,n}_{\delta}:=\frac{R^t[x]}{\langle x^{np^s}-\delta \rangle},$ where $\delta= \delta_0+u\delta_1+\dots+u^{t-1}\delta_{t-1}$ is a unit in $R^t$. For $n=1,\ 2, \ 3$ and $t=3$, we provide all types of ideals (constacyclic codes) and also give the torsional degrees as well as cardinalities of these codes.  

\medskip

			
\medskip
			

\end{abstract}

\section{Introduction}\label{section1}
Linear codes over fields constitute a fundamental class of codes in the theory of error-correcting codes. The landmark paper by Hammons et al. (\cite{hammons1994z}), where they showed that some good nonlinear codes can be constructed from linear codes over $\mathbb{Z}_4$, brought linear codes over rings to the forefront of coding theory research. Since then, researchers have explored various subclasses of linear codes over certain rings, highlighting their structural, theoretical, and practical significance. 
Among them, cyclic codes, in particular, and constacyclic codes, in general, represent important subclasses of linear codes, distinguished by their algebraic structure and the availability of efficient encoding and decoding algorithms, particularly those based on shift-register implementations. Both these subclasses have been extensively studied over fields as well as different classes of rings. One such class of rings is the rings of the types $\frac{\mathbb{F}_{p^m}[u]}{\langle u^t \rangle}=\mathbb{F}_{p^m}+u\mathbb{F}_{p^m}+\dots +u^{t-1}\mathbb{F}_{p^m}$, where $t$ is a positive integer. Over the last decade and a half, a lot of efforts have been made to the study of constacyclic codes of different lengths over these rings for $t=2,\ 3$, and 4 (for example, see \cite{chen2016constacyclic},  \cite{dinh2010constacyclic}, \cite{dinh2012repeated}, \cite{dinh2013structure}, \cite{dinh2013repeated}, \cite{dinh2017constacyclic}, \cite{dinh2020constacyclic}, \cite{liu2016repeated}, \cite{Sriwirach}). Efforts have also been made to give the structure of polycyclic codes over $\frac{\mathbb{F}_{p^m}[u]}{\langle u^t \rangle}$. In \cite{boudine2023classification}, results from valuation theory are used to give the generators of the ideals of the ring $\frac{\frac{\mathbb{F}_{p^m}[u]}{\langle u^t\rangle }[x]}{\langle f(x)^{p^s}\rangle}$, where $f(x)$ is a monic irreducible polynomial in $\frac{\mathbb{F}_{p^m}[u]}{\langle u^t\rangle }[x]$ and, as a particular case, for $f(x)=x^4+x^3+x^2+x+1$, $t=3$, and $p\equiv 2\,(\textnormal{mod}\,5)$ or $p\equiv 3\,(\textnormal{mod}\,5)$ certain parameters are computed. In \cite{AkaKanSar}, authors pointed out a gap in the proof of the result in \cite{boudine2023classification} giving the structure of the ideals and used basic commutative algebra techniques along with the principle of mathematical induction to give the structure of polycyclic codes over the ring $\frac{\mathbb{F}_{p^m}[u]}{\langle u^t \rangle}$ in a more general setting. Furthermore, as a particular case, the authors in \cite{AkaKanSar} computed parameters for a general irreducible polynomial $f(x)$ over $\mathbb{F}_{p^m}$ when $t=4$ and used these parameters to compute torsional degrees as well as cardinalities of these codes.
\par 
In this article, we continue this direction of research and study constacyclic codes of length $np^s$ over $\frac{\mathbb{F}_{p^m}[u]}{\langle u^t\rangle}$. We first give a special case of Theorem 3.2 in \cite{AkaKanSar} and give generators of all ideals of $\frac{\frac{\mathbb{F}_{p^m}[u]}{\langle u^t\rangle }[x]}{\langle x^{np^s}-\delta\rangle},$ where $\delta$ is a unit in  $\frac{\mathbb{F}_{p^m}[u]}{\langle u^t\rangle}$ and $(n,p)=1$ (Theorem \ref{Maintheorem}). We, then, give the modified form of the generators in the case $x^n-\delta_{0,0}$, where $\delta=\delta_0+u\delta_1 +\dots+u^{t-1}\delta_{t-1}$ and $\delta_{0,0}\in\mathbb{F}_{p^m}$ is such that  $\delta_{0,0}^{p^s}=\delta_0$, is irreducible over $\mathbb{F}_{p^m}$ (Theorem \ref{nps ideals}). We use this to give a complete description of ideals in the cases $n=1,2,$ and $3$ (Theorem \ref{nps ideals}, Theorem \ref{final theorem for section 4}, Theorem \ref{final theorem for section 5}). For the case $t=3$, we give all the types of constacyclic codes of length $p^s$, $2p^s$, and $3p^s$ and also give the torsional degrees as well as cardinalities of these codes, including the missing cases in the works of \cite{Sriwirach}.
We also describe a mild sufficient condition for any polynomial in $\mathbb{F}_{p^m}[x]$ of degree less than $n$ to be invertible in $\frac{R^t[x]}{\langle f(x)\rangle}$, where $f(x)\in R^t[x^{p^s}]$ (Proposition \ref{generaltheorem}), generalizing earlier results for polynomials of degree 1 and 2 for specific $f(x)$.

\section{\texorpdfstring{Constacyclic codes of length $np^s$ over $\frac{\mathbb{F}_{p^m}[u]}{\langle u^t\rangle}$}{}{}}\label{section2}
Recall that for a finite commutative ring $R$ with unity and a unit $\lambda \in R$, an $R$-submodule $C$ of $R^n$ is called a $\lambda$-constacyclic code of length $n$ over $R$ if $(\lambda c_{n-1}, c_0, c_1, \dots,c_{n-2}) \in C$ whenever $(c_0, c_1, \dots,c_{n-1}) \in C$  . Further, if we identify $n$-tuples with polynomials of degree $n-1$, then constacyclic codes are precisely the ideals of the ring $\frac{R[x]}{\langle x^n-\lambda \rangle}$. It is, therefore, enough to study ideals of the ring $\frac{R[x]}{\langle x^n-\lambda \rangle}$ to study $\lambda$-constacyclic codes over $R.$

Let $\mathbb{F}_{p^m}$, where $p$ is a prime, be the finite field of order $p^m$. For $t\geq1$, set $R^t:=\frac{\mathbb{F}_{p^m}[u]}{\langle u^t \rangle }.$ In this section, we give the structure of the ideals of $R^{t,n}_{\delta}:=\frac{R^t[x]}{\langle x^{np^s}-\delta\rangle}$, where  $n$ is a positive integer such that $(n,p)=1$ and $\delta$ is a unit in $R^t$. Writing $\delta=\delta_0+u\delta_1+\dots+u^{t-1}\delta_{t-1}$, where $\delta_i\in \mathbb{F}_{p^m},$ we note that $\delta$ is a unit in $R^t$ if and only if $\delta_0\ne 0$. Let $\delta_{0,0}=\delta_0^{p^{m-r}}$, where $s=mq+r$ and $0\leq r\leq m-1.$ Then $\delta_{0,0}^{p^s}=\delta_0.$ We first state a result from \cite{AkaKanSar} that will be helpful in obtaining the structure of the ideals of $R^{t,n}_{\delta}$.
\begin{proposition}(\cite{AkaKanSar}, Proposition 2.1) \label{FormofIdeals}
Let $A$ and $B$ be commutative rings with unity and let $\phi: A \rightarrow B$ be a surjective ring homomorphism with $ ker(\phi)=\langle \pi\rangle.$ If $I$ is an ideal of $A$ such that $\phi(I)$ is a finitely generated ideal then there exist $a_1,a_2,\dots,a_n$ in $I$ such that 
$$I=\langle a_1,a_2,\dots,a_n\rangle+\pi(I:\pi),$$
where $n$ is the number of generators of $\phi(I)$ and 
$$(I:\pi)=\{x\in R : x\pi \in I \}.$$
In particular, if $B$ is Noetherian then for every ideal $I$ of $A$ there exists a positive integer $n$ and $a_1,a_2,\dots,a_n$ in $I$ such that
$$I=\langle a_1,a_2,\dots,a_n\rangle+\pi(I:\pi).$$
\end{proposition}
We refer the reader to \cite{AkaKanSar} for the proof of this as well as the consequences of this proposition.
Let $x^n-\delta_{0,0}=f_1(x)^{n_1}f_2(x)^{n_2}\cdots f_l(x)^{n_l}$, where $l$ is a natural number and $n_i\geq 1$ for $1\leq i\leq l$, be the factorization of $x^n-\delta_{0,0}$ into irreducible polynomials over $\mathbb{F}_{p^m}$. Then every ideal of $\frac{\mathbb{F}_{p^m}[x]}{\langle (x^{n}-\delta_{0,0})^{p^s}\rangle}$ will be of the form $\langle f_1(x)^{k_1}f_2(x)^{k_2}\cdots f_l(x)^{k_l}\rangle$, where $0\leq k_i\leq n_ip^s$ and $1\leq i\leq l.$ Note that the map $\mu :R^{t,n}_{\delta}\rightarrow \frac{\mathbb{F}_{p^m}[x]}{\langle (x^{n}-\delta_{0,0})^{p^s}\rangle}$ such that $\mu(c(x))=c(x) (\textnormal{mod}\, u)$ is a surjective ring homomorphism with $\ker (\mu)=\langle u \rangle$. Also, the map $\Phi: R^{t,n}_{\delta} \rightarrow R^{t-1,n}_{\delta'}$, where $\delta'=\delta\textnormal{(mod }u^{t-1})\in R^{t-1}$,  given by $$\Phi(c(x))=c(x)\textnormal{ (mod } u^{t-1}),$$ is a surjective ring homomorphism. \\
With these notations, in our next theorem, we give the structure of the ideals of $R^{t,n}_{\delta}$. The approach of the proof is similar to that used in the proof of Theorem 3.2 in \cite{AkaKanSar}. For the sake of completeness, we state it here. As in the proof of Theorem 3.2 in \cite{AkaKanSar}, Proposition \ref{FormofIdeals} plays a critical role in the proof.

\begin{theorem}\label{Maintheorem}
    The ideals of the ring $R^{t,n}_{\delta}$ and their generators have one of the following forms.
    \begin{itemize}
        \item [(i)] Trivial ideals $\langle 0\rangle,$ $\langle 1 \rangle.$
        \item [(ii)] Any generator of a non-trivial ideal contained in $\langle u\rangle $ has the form:
        $$u^{(t-1)-i}(f_1(x)^{k_{1,i}}f_2(x)^{k_{2,i}}\cdots f_l(x)^{k_{l,i}})-u^{(t-1)-(i-1)}g(x),$$
        for some $0\leq i \leq t-2,\, g(x)\in R^{t,n}_{\delta},\,$ and $  0\leq k_{j,i}\leq n_jp^s$ (not all $k_{j,i}=n_jp^s$) for $1\leq j\leq l.$\\
        In fact, any non-trivial ideal $I$ contained in $\langle u \rangle$ has the form:
        \begin{align*}
           I=\langle& u^{(t-1)-i_1}(f_1(x)^{k_{1,i_1}}f_2(x)^{k_{2,i_1}}\cdots f_l(x)^{k_{l,i_1}})-u^{(t-1)-(i_1-1)}g_{i_1}(x),\dots, u^{(t-1)-i_d}\\&(f_1(x)^{k_{1,i_d}}
        f_2(x)^{k_{2,i_d}}\cdots f_l(x)^{k_{l,i_d}})-u^{(t-1)-(i_d-1)}g_{i_d}(x)\rangle, 
        \end{align*}
        where $0\leq i_1<i_2<\dots<i_d\leq t-2$, $g_{i_j}(x)\in R^{t,n}_{\delta}$ for $1\leq j\leq d.$
        \item[(iii)] Any non-trivial ideal not contained in $\langle u \rangle$ has the form:
        $$\langle (f_1(x)^{k_{1}}f_2(x)^{k_{2}}\cdots f_l(x)^{k_{l}})+ ur(x)\rangle+I,$$ 
        where $0\leq k_{i}\leq n_ip^s$ for $1\leq i\leq l$, $r(x)\in R^{t,n}_\delta$, and $I$ is an ideal as described in (ii).
        \end{itemize}
\end{theorem}
   Next, we give some results that will be helpful in giving the representation of elements of $R^{t,n}_\delta$ as well as in discussing the structure of constacyclic codes of length $np^s$ in more detail.
    \begin{proposition}\label{delta power n iff delta_0 power n}
    Let $\delta=\delta_0+u\delta_1+\dots+u^{t-1}\delta_{t-1}$, where $\delta_0,\delta_1,\dots,\delta_{t-1}\in\mathbb{F}_{p^m}$ be any element of $R^t$ and let $n$ be a natural number such that $p\nmid n$. Then $\delta$ is an $n^\textnormal{th}$ power in $R^t$ if and only if $\delta_0$ is an $n^\textnormal{th}$ power in $\mathbb{F}_{p^m}$.   
\end{proposition}
\begin{proof}If $\delta$ is an $n^\textnormal{th}$ power, then so is $\delta_0$ is clear. To prove the converse, let $\delta_0$ be an $n^\textnormal{th}$ power in $\mathbb{F}_{p^m}^*.$ Then there exists $\tilde{\delta_0}\in\mathbb{F}_{p^m}^*$ such that $\delta_0=\tilde{\delta_0}^n.$ We will use induction on $t$. For $t=1,$ the result holds trivially. Assume that the result holds for $t-1$. We will prove the result for $t$. Consider the map $\phi:R^t\rightarrow R^{t-1}$ such that $\phi(a)=a \,(\textnormal{mod}\, u^{t-1}).$ Note that $\phi(\delta)=\delta_0+u\delta_{1}+\dots +u^{t-2}\delta_{t-2}$ is an element of $R^{t-1}$ such that $\delta_0=\tilde{\delta_0}^n$ for some $\tilde{\delta_0}\in \mathbb{F}_{p^m}^*.$
   This implies that there exists $\beta=\beta_0+u\beta_1+\dots+u^{t-2}\beta_{t-2}$ such that $\phi(\delta)=\phi(\beta^n).$ Let $c_0$ be the coefficient of $u^{t-1}$ in $\beta^n$ and let $\tilde{\delta}=\beta + u^{t-1}c$ where $c=(n\beta_0)^{-1}(\delta_{t-1}-c_0) \in \mathbb{F}_{p^m}.$ To prove $\delta=\tilde{\delta}^n$, it is enough to see that the coefficient of $u^{t-1}$ in $\tilde{\delta}^n$ is  $\delta_{t-1}$ which follows since $n(t-1)\geq t$ for $t \geq n$. This proves the claim. \hfill{$\square$} 
\end{proof}
To prove our next result, we consider $f(x)=f_0+f_1x^{p^s}+f_2x^{2p^s}+\dots+f_{n}x^{np^s}\in R^t[x^{p^s}]$, where $f_i=f_{i,0}+uf_{i,1}+\dots+u^{t-1}f_{i,t-1}$ for $0\le i\le n$. Observe that for each $f_{i,j}$, where $0\le i\le n$ and $0\le j\le t-1$, there exist $\tilde{f}_{i,j}$ such that $\tilde{f}_{i,j}^{p^s}=f_{i,j}$.  With these notations, we give the following proposition.
\begin{proposition}\label{generaltheorem}
    If $\tilde{f}_{0,0}+\tilde{f}_{1,0}x+\dots+\tilde{f}_{n,0}x^n \in \mathbb{F}_{p^m}[x]$ is an irreducible polynomial of degree $n$, then every polynomial of degree less than  $n$ over $\mathbb{F}_{p^m}$ is invertible in the quotient ring $\frac{R^t[x]}{\langle f(x)\rangle}$.
\end{proposition}
\begin{proof}We first observe that  
\begin{align*}
f(x) - (\tilde{f}_{0,0}+\tilde{f}_{1,0}x+\dots+\tilde{f}_{n,0}x^{n})^{p^s} &=f(x)-(f_{0,0}+f_{1,0}x^{p^s}+\dots +f_{n,0}x^{np^s})\\&=u(\textnormal{some element of }R^t[x])
\end{align*}
Therefore, in  $\frac{R^t[x]}{\langle f(x)\rangle}$, 
$(\tilde{f}_{0,0}+\tilde{f}_{1,0}x+\dots+\tilde{f}_{n,0}x^{n})^{p^s}$ is a multiple of $u$. Let $g(x) \in \mathbb{F}_{p^m}[x]$ be a polynomial of degree less than $n$.
    Since $\tilde{f}_{0,0}+\tilde{f}_{1,0}x+\dots+\tilde{f}_{n,0}x^n$ is an irreducible polynomial of degree $n$, $\gcd(g(x),(\tilde{f}_{0,0}+\tilde{f}_{1,0}x+\dots+\tilde{f}_{n,0}x^n)^{p^s})=1$. Thus, there exists $a(x),b(x)\in \mathbb{F}_{p^m}[x]$ such that 
    $$a(x)g(x)+b(x)(\tilde{f}_{0,0}+\tilde{f}_{1,0}x+\dots+\tilde{f}_{n,0}x^n)^{p^s}=1.$$
Hence, $1- a(x)g(x)$ is a multiple of $u$ in $\frac{R^t[x]}{\langle f(x)\rangle}.$ Since $u$ is nilpotent, it follows that $a(x)g(x)$ and hence $g(x)$ is a unit.  \hfill{$\square$}
\end{proof}
 \begin{corollary}\label{deg n-1 polynomial in nps case}
 Let $p$ be a prime and $n$ be a natural number such that $p\nmid n$, and let $\alpha=\alpha_0+u\alpha_1+\dots+u^{t-1}\alpha_{t-1}\in R^t$ be such that $x^n-\alpha_{0,0}$ is irreducible over $\mathbb{F}_{p^m}$, where $\alpha_{0,0}^{p^s}=\alpha_0$. Then every non-zero polynomial $g(x)\in \mathbb{F}_{p^m}[x]$ having degree less than $n$ is invertible in  $\frac{R^t[x]}{\langle x^{np^s}-\alpha \rangle}$. In particular, any non-zero polynomial of the form $ax+b\in \mathbb{F}_{p^m}[x]$ is invertible in $\frac{R^t[x]}{\langle x^{np^s}-\alpha \rangle}$ whenever $n\geq 2$ and any non-zero polynomial of the form $ax^2+bx+c\in \mathbb{F}_{p^m}[x]$ is invertible  in $\frac{R^t[x]}{\langle x^{np^s}-\alpha \rangle}$ whenever $n\geq 3$.
 \end{corollary}
We remark that Proposition 6 of \cite{Sriwirach} follows as a special case of Corollary \ref{deg n-1 polynomial in nps case} by taking $n=2$ and $t=3$. Similarly, taking $n=3$ and $t=2$, Proposition 34 of \cite{dinh2020constacyclic} is obtained as a particular case of Corollary \ref{deg n-1 polynomial in nps case}.
\begin{corollary}\label{ax+b invertible in 2ps case}
   Let $\alpha=\alpha_0+u\alpha_1+\dots+u^{t-1}\alpha_{t-1}$ and $\beta=\beta_0+u\beta_1+\dots+u^{t-1}\beta_{t-1}\in R^t$ be such that $x^2+\alpha_{0,0}x+\beta_{0,0}$ is irreducible over $\mathbb{F}_{p^m}$, where $\alpha_{0,0}^{p^s}=\alpha_0$ and $\beta_{0,0}^{p^s}=\beta_0$. Then any polynomial of the form $ax+b\in \mathbb{F}_{p^m}[x]$, where $a$ and $b$ are not simultaneously zero, is invertible in  $\frac{R^t[x]}{\langle x^{2p^s}+\alpha x^{p^s}+\beta\rangle}$.
\end{corollary}

We remark that if we take $\beta = \alpha^2$ in Corollary \ref{ax+b invertible in 2ps case}, we get Lemma 15 of \cite{dinh2020constacyclic} as a particular case.
As a consequence of Proposition \ref{generaltheorem}, we can give the representation of any element of $\frac{R^t[x]}{\langle f(x) \rangle}$ and also the condition for any element to be a unit. We state this in the following remark.
  \begin{remark}\label{generalrepresentation}
   Using the notation and conditions of Proposition \ref{generaltheorem}, we have the following.
    \begin{itemize}
        \item[(i)] An arbitrary element $c(x)$ of $\frac{R^t[x]}{\langle f(x) \rangle}$ can be uniquely written as
        \begin{align*}
            c(x)=\underset{j=0}{\overset{p^s-1}{\sum}}&\big(\underset{i=0}{\overset{n-1}{\sum}}{c_{i,j}^{0}x^i}\big)(\tilde{f}_{0,0}+\tilde{f}_{1,0}x+\dots+\tilde{f}_{n,0}x^n)^j+u\underset{j=0}{\overset{p^s-1}{\sum}}\big(\underset{i=0}{\overset{n-1}{\sum}}{c_{i,j}^{1}x^i}\big)(\tilde{f}_{0,0}+\tilde{f}_{1,0}x+\dots\\&+\tilde{f}_{n,0}x^n)^j+\dots+u^{t-1}\underset{j=0}{\overset{p^s-1}{\sum}}\big(\underset{i=0}{\overset{n-1}{\sum}}{c_{i,j}^{t-1}x^i}\big)(\tilde{f}_{0,0}+\tilde{f}_{1,0}x+\dots+\tilde{f}_{n,0}x^n)^j,
        \end{align*}
        where $c_{i,j}^k\in \mathbb{F}_{p^m}$ for $0\leq j\leq p^s-1$, $0\leq i\leq n-1$, and $0\leq k\leq t-1.$
        \item[(ii)] An element $c(x),$ as given above, is a non-unit if and only if $c_{i,0}^0=0$ for all $0\le i\le n-1$.
    \end{itemize}
\end{remark}
Now, let $k$ be the smallest integer such that $1\leq k\leq t-1$, and $\delta_k \neq0$. Then $(x^n-\delta_{0,0})^{p^s}=x^{np^s}-\delta_{0,0}^{p^s}=\delta -\delta_0=u^k\delta_k+u^{k+1}\delta_{k+1}+\dots+u^{t-1}\delta_{t-1}=u^k(\delta_k+u\delta_{k+1}+\dots+u^{t-1-k}\delta_{t-1}).$ Since $\delta_k\neq 0$, we have $\langle (x^n-\delta_{0,0})^{p^s} \rangle=\langle u^k\rangle$. Moreover, $x^n-\delta_{0,0}$ is nilpotent in $R^{t,n}_{\delta}$ with nilpotency index $\lceil \frac{t}{k}\rceil p^s$. Also, if $x^n-\delta_{0,0}$ is irreducible over $\mathbb{F}_{p^m}$, then, as in  Remark \ref{generalrepresentation}, an arbitrary element $c(x)$ of $R^{t,n}_{\delta}$ can be uniquely written as
        \begin{align*}
            c(x)=\underset{j=0}{\overset{p^s-1}{\sum}}\big(\underset{i=0}{\overset{n-1}{\sum}}{c_{i,j}^{0}x^i}\big)(x^n-\delta_{0,0})^j+u\underset{j=0}{\overset{p^s-1}{\sum}}&\big(\underset{i=0}{\overset{n-1}{\sum}}{c_{i,j}^{1}x^i}\big)(x^n-\delta_{0,0})^j+\dots+\\&u^{t-1}\underset{j=0}{\overset{p^s-1}{\sum}}\big(\underset{i=0}{\overset{n-1}{\sum}}{c_{i,j}^{t-1}x^i}\big)(x^n-\delta_{0,0})^j,
        \end{align*}
        where $c_{i,j}^k\in \mathbb{F}_{p^m}$ for $0\leq j\leq p^s-1,\,0\leq i\leq n-1,$ and $0\leq k\leq t-1$. Further, $c(x)$ is a non-unit if and only if $c_{i,0}^0=0$ for all $0\leq i\leq n-1$. \\
        Using this representation of an element of $R^{t,n}_\delta$, we can write the generator $u^{(t-1)-i}(x^n-\delta_{0,0})^{a_i}-u^{(t-1)-(i-1)}g(x)$ as 
$$u^{(t-1)-i}(x^n-\delta_{0,0})^{a_{i}}+u^{(t-1)-(i-1)}(x^n-\delta_{0,0})^{t_{i-1,0}}g_{i-1,0}(x)+\dots + u^{(t-1)}(x^n-\delta_{0,0})^{t_{i-1,i-1}}g_{i-1,i-1}(x),$$
where each $g_{i-1,j}(x)$ is either $0$ or a unit in $\frac{\mathbb{F}_{p^m}[x]}{\langle x^{np^s}-\delta_0\rangle}$  for $0\leq j\leq i-1$ and $p^s-1\geq a_i>t_{i-1,0}>t_{i-1,1}>\dots>t_{i-1,i-1}\geq 0$. Thus we give the ideals of $R^{t,n}_{\delta}$ in the special case when $x^n-\delta_{0,0}$ is an irreducible polynomial over $\mathbb{F}_{p^m}$ in the following theorem.
\begin{theorem}\label{nps ideals}
    Let $\delta=\delta_0+u\delta_1+\dots+u^{t-1}\delta_{t-1}\in R^t$ be a unit and let $\delta_{0, 0}^{p^s}=\delta_0$. If $x^n-\delta_{0,0}$ is an irreducible polynomial over $\mathbb{F}_{p^m}$ then the ideals of the ring $R^{t,n}_{\delta}$ have one of the following forms.
    \begin{itemize}
        \item [(i)] Trivial ideals $\langle 0\rangle,$ $\langle 1 \rangle.$
        \item [(ii)] Any generator of a non-trivial ideal contained in $\langle u\rangle $ has the form
        \begin{align*}
        \theta_i(u, x^n-\delta_{0,0})&=u^{(t-1)-i}(x^n-\delta_{0,0})^{a_{i}}+u^{(t-1)-(i-1)}(x^n-\delta_{0,0})^{t_{i-1,0}}g_{i-1,0}(x)+\dots +\\&\,\,\,\,\,\, u^{(t-1)}(x^n-\delta_{0,0})^{t_{i-1,i-1}}g_{i-1,i-1}(x),
        \end{align*}
        for some $0\leq i\leq t-2$ and each $g_{i-1,j}(x)$ is either 0 or a unit in $\frac{\mathbb{F}_{p^m}[x]}{\langle x^{np^s}-\delta_0\rangle}$ for $0\leq j\leq i-1$ and $p^s-1\geq a_i>t_{i-1,0}>t_{i-1,1}>\dots>t_{i-1,i-1}\geq 0$.\\
        In fact, any non-trivial ideal $I$ contained in $\langle u \rangle$ has the form:
        $$\langle \theta_{i_1}(u, x^n-\delta_{0,0}),\dots, \theta_{i_d}(u, x^n-\delta_{0,0})\rangle,$$ 
         where $0\leq i_1<i_2<\dots<i_d\leq t-2$, $0\leq a_{i_1}<a_{i_2}<\dots<a_{i_d}\leq p^s-1$, $0\leq t_{(i_1-1),(i_1-1)}<t_{(i_1-1),(i_1-2)}<\dots<t_{(i_1-1),0}<a_{i_1},\, 0\leq t_{(i_2-1),(i_2-1)}<t_{(i_2-1),(i_2-2)}<\dots<t_{(i_2-1),0}<a_{i_2},\, \dots,\,$ and $0\leq t_{(i_d-1),(i_d-1)}<t_{(i_d-1),(i_d-2)}<\dots<t_{(i_d-1),0}<a_{i_d}.$  
        \item[(iii)] Any non-trivial ideal not contained in $\langle u \rangle$ has the form:
        $$\langle (x^n-\delta_{0,0})^a+ uf(x)\rangle+I,$$ where $f(x)\in R^{t,n}_{\delta}$, $I$ is an ideal as described in (ii), and $0\leq a_{i_d}<a\leq p^s-1$.
        \end{itemize}
\end{theorem}
As a special case, we have the following corollary.
\begin{corollary}\label{chainfornpscase} Let $\delta=\delta_0+u\delta_1+\dots+u^{t-1}\delta_{t-1} \in R^t$ be such that $x^n-\delta_{0,0}$ is an irreducible polynomial over $\mathbb{F}_{p^m}$. Then, the ideals of the ring $R^{t,n}_{\delta}$ form a chain if and only if $\delta_1\neq 0$. In fact, the ideals  of $R^{t,n}_{\delta}$ are
$$\langle 0\rangle\subset \langle (x^n-\delta_{0,0})^{tp^s-1}\rangle \subset \langle (x^n-\delta_{0,0})^{tp^s-2}\rangle \subset \dots \subset \langle x^n-\delta_{0,0}\rangle \subset \langle 1\rangle.$$ 
Moreover, the number of codewords in $\langle (x^n-\delta_{0,0})^i\rangle$ is $p^{nm(tp^s-i)}$ for $0\leq i\leq tp^s.$
\end{corollary}
\begin{proof}
    If $\delta_1\ne 0$, then from Theorem \ref{nps ideals}, ideals form a chain. Conversely, assume that $\delta_1=0.$ We will show that $\langle u,x^n-\delta_{0,0}\rangle$ is a non-principal ideal. Let, if possible, $u\in \langle x^n-\delta_{0,0}\rangle,$ then there exist $m_1(x),m_2(x)\in R^t[x]$ such that $u=(x^n-\delta_{0,0})m_1(x)+(x^{np^s}-\delta)m_2(x)$. Let $\gamma$ be a root of $x^n-\delta_{0,0}$ in some extension of $\mathbb{F}_{p^m}$. Then, substituting $\gamma$ in the above equation, we have $u=(\delta_0-\delta)m_2(\gamma)\implies u\in \langle u^2\rangle,$ a contradiction. Now, if possible, assume that $x^n-\delta_{0,0}\in\langle u\rangle.$ Then, there exist $n_1(x),n_2(x)\in R^t[x]$ such that $x^n-\delta_{0,0}=un_1(x)+(x^{np^s}-\delta)n_2(x).$ Taking modulo $u$ both sides we get $x^n-\delta_{0,0}=(x^{n}-\delta_{0,0})^{p^s}(n_2(x) \textnormal{ mod }u)$, which leads to a contradiction.\hfill{$\square$}
\end{proof}
We remark that Proposition 4.2 in \cite{dinh2016repeated} and Theorem 3.5 in \cite{consta2psover} are special cases of Corollary \ref{chainfornpscase}. Also, Corollary \ref{chainfornpscase} in the case $n=1$ was earlier proved for a finite chain ring in \cite{sharma2018structure}.

\section{Ideals of \texorpdfstring{$R^{t,1}_{\delta}$}{}} \label{section3}
In this section, we use the structure of ideals of $R^{t,n}_{\delta}:=\frac{R^t [x]}{\langle x^{np^s}-\delta\rangle},$ where $\delta = \delta_0+u\delta_1+\dots+u^{t-1}\delta_{t-1}$ is a unit in $R^t:=\frac{\mathbb{F}_{p^m}[u]}{\langle u^t \rangle}$ and $\delta_i\in \mathbb{F}_{p^m}$ for $0\leq i \leq t-1$ given in Section \ref{section2} to discuss the case when $n=1$, that is, to discuss the ideals of $R^{t,1}_{\delta}:=\frac{R^t [x]}{\langle x^{p^s}-\delta\rangle}$. We first note that, in this case, $x^n-\delta_{0,0}=x-\delta_{0,0}$, where $\delta_{0,0}^{p^s}=\delta_0$, is irreducible. Hence, taking $n=1$ in Theorem \ref{nps ideals}, we get the generators and the structure of ideals of $R^{t,1}_{\delta}$, in this case.\\
For clarity and better understanding, we remark that, as in Section \ref{section2}, if $k$ is the smallest integer such that $\delta_k \neq0$, where $1\leq k\leq t-1$, then $\langle (x-\delta_{0,0})^{p^s} \rangle=\langle u^k\rangle$, $x-\delta_{0,0}$ is nilpotent in $R^{t,1}_{\delta}$ with nilpotency index $\lceil \frac{t}{k}\rceil p^s$, and an arbitrary element $c(x)$ of $R^{t,1}_{\delta}$ can be uniquely written as
        \begin{align*}
            c&(x)=\underset{j=0}{\overset{p^s-1}{\sum}}c_j^{0}(x-\delta_{0,0})^j+u\underset{j=0}{\overset{p^s-1}{\sum}}c_j^{1}(x-\delta_{0,0})^j+\dots+u^{t-1}\underset{j=0}{\overset{p^s-1}{\sum}}c_{j}^{t-1}(x-\delta_{0,0})^j,
        \end{align*}
        where $c_{j}^k\in \mathbb{F}_{p^m}$ for $0\leq j\leq p^s-1$ and $0\leq k\leq t-1$. Moreover, $c(x)$ is a non-unit if and only if $c_{0}^0=0$. Also, as in  Corollary \ref{chainfornpscase}, the ideals of $R^{t,1}_{\delta}$ form a chain if and only if $\delta_1\ne 0$. In fact, the ideals  of $R^{t,1}_{\delta}$, in this case, are
    $$\langle 0\rangle \subset \langle (x-\delta_{0,0})^{tp^s-1}\rangle\subset \langle (x-\delta_{0,0})^{tp^s-2}\rangle\subset \dots \subset \langle (x-\delta_{0,0})\rangle\subset \langle 1\rangle.$$
    Moreover, for $0 \leq i \leq tp^s$, the cardinality of $\langle (x-\delta_{0,0})^{i}\rangle=p^{m(tp^s-i)}.$
Now, let $t=3$ and let $\delta=\delta_0+u\delta_1+u^2\delta_2 \in R^3$ be a unit. We will discuss the ideals $R^{3,1}_{\delta}$ in this case and also give the cardinality of these ideals. If $\delta_1\neq 0$ then, as described above and also  by Corollary \ref{chainfornpscase}, the ideals $\langle (x-\delta_{0,0})^{i}\rangle$ ($0 \leq i \leq 3p^s$) of $R^{3,1}_{\delta}$ form a chain and cardinality of $\langle (x-\delta_{0,0})^{i}\rangle$ is $p^{m(3p^s-i)}$. If $\delta_1=0=\delta_2$ then the ideals of $R^{3,1}_{\delta}$ are same as the ideals of $\frac{R^3[x]}{\langle (x-\delta_{0,0})^{p^s}\rangle}.$ We refer the reader to \cite{hesari2024torsion} and \cite{AkaKanSar} for the ideals of $R^{3,1}_{\delta}$ and the cardinality of these ideals in this case. Finally, let $\delta_1= 0$ and $\delta_2\neq 0$. We now discuss the cardinality of the ideals of $R^{3,1}_{\delta}$ in this case. First, we list all the different types of ideals of the ring $R^{3,1}_{\delta}$ in this case.
\begin{theorem}\label{ideals of ps specific case}
    The ideals of the ring $\frac{R^3[x]}{\langle x^{p^s}-\delta_0-u^2\delta_2\rangle}$ have one of the following eight types:
     \begin{enumerate}
           \item{\label{Type 1'}} $\langle 0 \rangle,\, \langle 1 \rangle.$
           \item{\label{Type 2'}} $\langle u^2 (x-\delta_{0,0})^a \rangle$ where $0\leq a \leq p^s-1.$
           \item{\label{Type 3'}} $\langle u(x-\delta_{0,0})^a+u^2(x-\delta_{0,0})^{t}h(x)\rangle$ where $0\leq t<L< a\leq p^s-1,$ $h(x)$ is either $0$ or a unit in $\frac{\mathbb{F}_{p^m}[x]}{\langle (x-\delta_{0,0})^{p^s}\rangle},$ and $L$ is the smallest non-negative integer such that $u^2(x-\delta_{0,0})^L\in \langle u(x-\delta_{0,0})^a+u^2(x-\delta_{0,0})^{t}h(x)\rangle.$
           
           \item{\label{Type 4'}}$\langle u(x-\delta_{0,0})^a+u^2(x-\delta_{0,0})^{t}h(x),\, u^2(x-\delta_{0,0})^b\rangle$ where $0\leq t<b<L< a\leq p^s-1$, $h(x)$ is either $0$ or a unit in $\frac{\mathbb{F}_{p^m}[x]}{\langle (x-\delta_{0,0})^{p^s}\rangle},$ and $L$ is the smallest non-negative integer such that $u^2(x-\delta_{0,0})^{L}\in \langle u(x-\delta_{0,0})^a+u^2(x-\delta_{0,0})^{t}h(x)\rangle.$
           
           \item{\label{Type 5'}} $\langle (x-\delta_{0,0})^{a}+u(x-\delta_{0,0})^{t_0}h_0(x)+u^2(x-\delta_{0,0})^{t_1}h_1(x)\rangle$ where $0\leq t_1<t_0<L< a\leq p^s-1,\,0\leq t_1<M<L,\,$ $h_i(x)$ is either $0$ or a unit in $\frac{\mathbb{F}_{p^m}[x]}{\langle (x-\delta_{0,0})^{p^s}\rangle}$ for $i=0,\,1$ and $L,M$ are the smallest non-negative integers such that $u(x-\delta_{0,0})^L+u^2g(x),\, u^2(x-\delta_{0,0})^M\in \langle (x-\delta_{0,0})^{a}+u(x-\delta_{0,0})^{t_0}h_0(x)+u^2(x-\delta_{0,0})^{t_1}h_1(x)\rangle$ for some $g(x)\in \frac{\mathbb{F}_{p^m}[x]}{\langle (x-\delta_{0,0})^{p^s}\rangle}.$

           \item{\label{Type 6'}}$\langle (x-\delta_{0,0})^{a}+u(x-\delta_{0,0})^{t_0}h_0(x)+u^2(x-\delta_{0,0})^{t_1}h_1(x),\,u^2 (x-\delta_{0,0})^b \rangle\rangle$ for $0\leq t_1<t_0< a\leq p^s-1,\, 0\leq t_1<b <L<a,\,0\leq t_0<M<L$, $h_i(x)$ is either $0$ or a unit in $\frac{\mathbb{F}_{p^m}[x]}{\langle (x-\delta_{0,0})^{p^s}\rangle}$ for $i=0,\,1,$ and $L,M$ are the smallest non-negative integers such that $u^2(x-\delta_{0,0})^L,\, u(x-\delta_{0,0})^M+u^2g(x)\in \langle (x-\delta_{0,0})^{a}+u(x-\delta_{0,0})^{t_0}h_0(x)+u^2(x-\delta_{0,0})^{t_1}h_1(x) \rangle$ for some $g(x)\in \frac{\mathbb{F}_{p^m}[x]}{\langle (x-\delta_{0,0})^{p^s}\rangle}.$

           \item{\label{Type 7'}} $\langle (x-\delta_{0,0})^{a}+u(x-\delta_{0,0})^{t_0}h_0(x)+u^2(x-\delta_{0,0})^{t_1}h_1(x),\, u(x-\delta_{0,0})^b+u^2(x-\delta_{0,0})^{t_2}h_2(x)\rangle$ where $0\leq t_1<t_0<b<L< a\leq p^s-1\,, 0\leq t_2<b,\,0\leq t_i<M<b$ for $i=1,2$, $h_i(x)$ is either $0$ or a unit in $\frac{\mathbb{F}_{p^m}[x]}{\langle (x-\delta_{0,0})^{p^s}\rangle}$ for $0\leq i\leq 2,$ $L,M$ are the smallest non-negative integers such that $u(x-\delta_{0,0})^L+u^2g(x)\in \langle (x-\delta_{0,0})^{a}+u(x-\delta_{0,0})^{t_0}h_0(x)+u^2(x-\delta_{0,0})^{t_1}h_1(x)\rangle$ and $u^2(x-\delta_{0,0})^M\in \langle (x-\delta_{0,0})^{a}+u(x-\delta_{0,0})^{t_0}h_0(x)+u^2(x-\delta_{0,0})^{t_1}h_1(x),\, u(x-\delta_{0,0})^b+u^2(x-\delta_{0,0})^{t_2}h_2(x)\rangle$ for some $g(x)\in \frac{\mathbb{F}_{p^m}[x]}{\langle (x-\delta_{0,0})^{p^s}\rangle}.$

           \item{\label{Type 8'}}$\langle (x-\delta_{0,0})^{a}+u(x-\delta_{0,0})^{t_0}h_0(x)+u^2(x-\delta_{0,0})^{t_1}h_1(x),\, u(x-\delta_{0,0})^b+u^2(x-\delta_{0,0})^{t_2}h_2(x),\, u^2(x-\delta_{0,0})^c\rangle$ where $0\leq t_1<t_0< a\leq p^s-1,\, 0\leq t_1,t_2<c<b<a,\,0\leq t_0<b$,  $h_i(x)$ is either $0$ or a unit in $\frac{\mathbb{F}_{p^m}[x]}{\langle (x-\delta_{0,0})^{p^s}\rangle}$ for $0\leq i\leq 2$, $c<M$, the smallest non-negative integer such that $u^2(x-\delta_{0,0})^M\in \langle (x-\delta_{0,0})^{a}+u(x-\delta_{0,0})^{t_0}h_0(x)+u^2(x-\delta_{0,0})^{t_1}h_1(x),\, u(x-\delta_{0,0})^b+u^2(x-\delta_{0,0})^{t_2}h_2(x)\rangle$, and $b<L$, the smallest non-negative integer such that $u(x-\delta_{0,0})^L+u^2g(x)\in \langle (x-\delta_{0,0})^{a}+u(x-\delta_{0,0})^{t_0}h_0(x)+u^2(x-\delta_{0,0})^{t_1}h_1(x)\rangle$ for some $g(x)\in \frac{\mathbb{F}_{p^m}[x]}{\langle (x-\delta_{0,0})^{p^s}\rangle}.$
       \end{enumerate}
\end{theorem}
In the following three lemmas, we compute the parameters that help us in computing the $i^{\textnormal{th}}$ torsional degree and consequently the cardinality of the ideals of $R^{3,1}_{\delta}$ when  $\delta_1= 0$ and $\delta_2\neq 0$.
\begin{lemma}\label{first parameter computation}
        Let $L$ be the smallest non-negative integer such that $u^2(x-\delta_{0,0})^L\in \langle u(x-\delta_{0,0})^a+u^2(x-\delta_{0,0})^{t}h(x)\rangle,$ where $h(x)$, if non-zero, is a unit in $\frac{\mathbb{F}_{p^m}[x]}{\langle (x-\delta_{0,0})^{p^s}\rangle}$. Then
        $$L= \begin{cases} 
             a & \textnormal{ if } h(x)=0, \\
             \Min\{a,\,p^s-a+t\}  & \textnormal{ if } h(x) \ne 0. \\
   \end{cases}$$
\end{lemma}
\begin{proof}
Let $C= \langle u(x-\delta_{0,0})^a+u^2(x-\delta_{0,0})^{t}h(x)\rangle$. If $h(x)=0$ then $L=a$ as $u^2(x-\delta_{0,0})^a \in C$. Now, let  $h(x)\ne 0$ and let $\alpha$ be a non-negative integer such that $u^2(x-\delta_{0,0})^{\alpha}\in C$. Then there exist  $c_i(x)=\underset{j=0}{\overset{p^s-1}{\sum}}a_{i,j}(x-\delta_{0,0})^j\in \frac{\mathbb{F}_{p^m}[x]}{\langle (x-\delta_{0,0})^{p^s}\rangle}$ for $i=0,1$ such that
\begin{align*}
  u^2(x-\delta_{0,0})^{\alpha}&=u(x-\delta_{0,0})^{a}c_0(x)+u^2(x-\delta_{0,0})^{a}c_1(x)+u^2(x-\delta_{0,0})^{t}h(x)c_0(x)\\& 
  =u(x-\delta_{0,0})^{a}\underset{j=0}{\overset{p^s-1}{\sum}}a_{0,j}(x-\delta_{0,0})^j+u^2(x-\delta_{0,0})^{a}c_1(x)\\&+u^2(x-\delta_{0,0})^{t}h(x)\underset{j=0}{\overset{p^s-1}{\sum}}a_{0,j}(x-\delta_{0,0})^j.
  \end{align*}Since $\langle(x-\delta_{0,0})^{p^s}\rangle=\langle u^2\rangle$, we have
  \begin{align*}
  u^2(x-\delta_{0,0})^{\alpha}&=
  u(x-\delta_{0,0})^{a}\underset{j=0}{\overset{p^s-1-a}{\sum}}a_{0,j}(x-\delta_{0,0})^j+u^2(x-\delta_{0,0})^{a}c_1(x)\\&+u^2(x-\delta_{0,0})^{t}h(x)\underset{j=0}{\overset{p^s-1-a}{\sum}}a_{0,j}(x-\delta_{0,0})^j\\&
  =u(x-\delta_{0,0})^{a}\underset{j=0}{\overset{p^s-1-a}{\sum}}a_{0,j}(x-\delta_{0,0})^j+u^2(x-\delta_{0,0})^{a}c_1(x)\\&
  +u^2(x-\delta_{0,0})^{p^s-a+t}h(x)\underset{j=0}{\overset{a-1}{\sum}}a_{0,j+p^s-a}(x-\delta_{0,0})^j.
\end{align*}
Let $\underset{j=0}{\overset{p^s-1-a}{\sum}}a_{0,j}(x-\delta_{0,0})^j=c_{0,0}(x)$ and $\underset{j=0}{\overset{a-1}{\sum}}a_{0,{j+p^s-a}}(x-\delta_{0,0})^j=c_{0,1}(x).$ Then $c_0(x)=c_{0,0}(x)+(x-\delta_{0,0})^{p^s-a}c_{0,1}(x).$ Furthermore, we have 
\begin{align}
&(x-\delta_{0,0})^a c_{0,0}(x) = 0, \tag{1}\\
&(x-\delta_{0,0})^{a}c_1(x) + (x-\delta_{0,0})^{p^s-a+t} h(x)c_{0,1}(x)=(x-\delta_{0,0})^{\alpha}. \tag{2}
\end{align}
Equation (1) gives $c_{0,0}(x)=0$. Using this in Equation (2), we get\\
\[(x-\delta_{0,0})^ac_1(x)+(x-\delta_{0,0})^{p^s-a+t}h(x)c_{0,1}(x)=(x-\delta_{0,0})^{\alpha}.\]
Thus $\alpha \geq \min\{a,\, p^s-a+t\}.$ In particular, since $u^2(x-\delta_{0,0})^{L}\in C$, we have $L\geq \min\{a,\, p^s-a+t\}.$ Also, if we take $c_1(x)=1,\, c_{0,1}(x)=0$, we get $u^2(x-\delta_{0,0})^a\in C$ and if we take $c_1(x)=0,\, c_{0,1}(x)=h(x)^{-1}$, we get $u^2(x-\delta_{0,0})^{p^s-a+t}\in C.$ Since $L$ is the smallest non-negative integer satisfying $u^2(x-\delta_{0,0})^{L}\in C $, we have  $L\leq \min\{a,\, p^s-a+t\}$ and hence $L=\min\{a,\,p^s-a+t\}.$\hfill $\square$
\end{proof}    
\begin{lemma}\label{second parameter computation}
        Let $L$ be the smallest non-negative integer such that $u^2(x-\delta_{0,0})^L\in \langle (x-\delta_{0,0})^{a}+u(x-\delta_{0,0})^{t_0}h_0(x)+u^2(x-\delta_{0,0})^{t_1}h_1(x)\rangle$, where for $i=0$ and $1$, $h_i(x)$, if non-zero, is a unit in $\frac{\mathbb{F}_{p^m}[x]}{\langle (x-\delta_{0,0})^{p^s}\rangle}.$ Then
        $$L= \begin{cases} 
             0 & \textnormal{ if } h_0(x)=0, \\
             0&\textnormal{ if } h_0(x)\neq 0 \textnormal{ and } a<p^s-a+t_0,\\
             \Min\{a,\,\beta\}  & \textnormal{ if } h_0(x) \ne 0 \textnormal{ and } a\geq p^s-a+t_0, \\
   \end{cases}$$
   where $\beta:=\max\{k:(x-\delta_{0,0})^k\mid ((x-\delta_{0,0})^{t_0}h_0(x) -h_0(x)^{-1}(x-\delta_{0,0})^{2a-p^s-t_0}-h_1(x)h_0(x)^{-1}(x-\delta_{0,0})^{a+t_1-t_0}h_1(x)h_0(x)^{-1})\}.$
\end{lemma}
\begin{proof}
Let $C=\langle(x-\delta_{0,0})^{a}+u(x-\delta_{0,0})^{t_0}h_0(x)+u^2(x-\delta_{0,0})^{t_1}h_1(x)\rangle$ and let $\alpha$ be a non-negative integer such that $u^2(x-\delta_{0,0})^{\alpha} \in C$. Then there exist $c_i(x)=\underset{j=0}{\overset{p^s-1}{\sum}}a_{i,j}(x-\delta_{0,0})^j\in\frac{\mathbb{F}_{p^m}[x]}{\langle (x-\delta_{0,0})^{p^s}\rangle}$ for $i=0,1,2$ such that
\begin{align*}
  u^2(x-\delta_{0,0})^{\alpha}&=(x-\delta_{0,0})^a c_0(x) +u\bigl((x-\delta_{0,0})^ac_1(x)+ (x-\delta_{0,0})^{t_0}h_0(x)c_0(x)\bigr)+\\&
  u^2\bigl((x-\delta_{0,0})^ac_2(x)+(x-\delta_{0,0})^{t_0}h_0(x)c_1(x) +(x-\delta_{0,0})^{t_1}h_1(x)c_0(x)\bigr).
  \end{align*}
   Let $\underset{j=0}{\overset{p^s-1-a}{\sum}}(a_{0,j}+b_{0,j}x)(x-\delta_{0,0})^j=c_{0,0}(x)$ and $\underset{j=0}{\overset{a-1}{\sum}}(a_{0,{j+p^s-a}}+b_{0,{j+p^s-a}}x)(x-\delta_{0,0})^j=c_{0,1}(x).$ Then $c_0(x)=c_{0,0}(x)+(x-\delta_{0,0})^{p^s-a}c_{0,1}(x).$ Since $\langle(x-\delta_{0,0})^{p^s}\rangle=\langle u^2\rangle$, we have
   \begin{align*}
  u^2(x-\delta_{0,0})^{\alpha}&=(x-\delta_{0,0})^a c_{0,0}(x) +u\bigl((x-\delta_{0,0})^ac_1(x)+ (x-\delta_{0,0})^{t_0}h_0(x)(c_{0,0}(x)+\\&(x-\delta_{0,0})^{p^s-a}c_{0,1}(x))\bigr)+
  u^2\bigl((x-\delta_{0,0})^ac_2(x)+c_{0,1}(x)+(x-\delta_{0,0})^{t_0}h_0(x)c_1(x)\\& +(x-\delta_{0,0})^{t_1}h_1(x)(c_{0,0}(x)+(x-\delta_{0,0})^{p^s-a}c_{0,1}(x))\bigr)
\end{align*}
and hence 
\begin{align*}
&(x-\delta_{0,0})^a c_{0,0}(x) = 0, \tag{1}\\
&(x-\delta_{0,0})^ac_1(x)+ (x-\delta_{0,0})^{t_0}h_0(x)(c_{0,0}(x)+(x-\delta_{0,0})^{p^s-a}c_{0,1}(x)) = 0, \tag{2}\\
&(x-\delta_{0,0})^ac_2(x)+c_{0,1}(x)+(x-\delta_{0,0})^{t_0}h_0(x)c_1(x) +(x-\delta_{0,0})^{t_1}h_1(x)(c_{0,0}(x)+(x-\delta_{0,0})^{p^s-a}\\&c_{0,1}(x))=(x-\delta_{0,0})^{\alpha}. \tag{3}
\end{align*}

Equation (1) gives $c_{0,0}(x) = 0$. 
Using this in Equations (2) and (3), we get
\begin{align*}
&(x-\delta_{0,0})^ac_1(x)+ (x-\delta_{0,0})^{p^s-a+t_0}h_0(x)c_{0,1}(x) = 0, \tag{4}\\
&(x-\delta_{0,0})^ac_2(x)+c_{0,1}(x)+(x-\delta_{0,0})^{t_0}h_0(x)c_1(x) +(x-\delta_{0,0})^{t_1}h_1(x)(x-\delta_{0,0})^{p^s-a}c_{0,1}(x)\\&=(x-\delta_{0,0})^{\alpha}.\tag{5}
\end{align*}
If $h_0(x)= 0$ then $((x-\delta_{0,0})^{a}+u^2(x-\delta_{0,0})^{t_1}h_1(x))(x-\delta_{0,0})^{p^s-a}\in C$, which gives $u^2(1+(x-\delta_{0,0})^{p^s-a+t_1}h_1(x))\in C$. Since $p^s-a+t_1 \geq 1,$ we have $u^2\in C$ and hence $L=0$.\\
So, let $h_0(x) \neq 0$.
We will consider two cases.\\
\textbf{Case 1.}  $a < p^s-a+t_0$.\\ In this case, we have $((x-\delta_{0,0})^{a}+u^2(x-\delta_{0,0})^{t_1}h_1(x))((x-\delta_{0,0})^{p^s-a}-u(x-\delta_{0,0})^{p^s-2a+t_0}h_0(x))\in C,$ which gives $u^2\{1+(x-\delta_{0,0})^{p^s-a+t_1}h_1(x)-(x-\delta_{0,0})^{p^s-2a+2t_0}h_0(x)^2\}\in C.$ Since $p^s-2a+2t_0>0$ and $p^s-a+t_1>0,$ we get $u^2\in C$ and hence, once again, $L=0.$
\\
\textbf{Case 2.}  $a \geq p^s-a+t_0$.\\
In this case, Equation (4) gives
\[{c}_{0,1}(x)=-h_0(x)^{-1}(x-\delta_{0,0})^{2a-p^s-t_0}c_1(x).\]Using this in Equation (5), we get 
\begin{align*}
&(x-\delta_{0,0})^ac_2(x)+\{(x-\delta_{0,0})^{t_0}h_0(x) -h_0(x)^{-1}(x-\delta_{0,0})^{2a-p^s-t_0}-h_1(x)h_0(x)^{-1}\\&(x-\delta_{0,0})^{a+t_1-t_0}h_1(x)h_0(x)^{-1}\}c_1(x)=(x-\delta_{0,0})^{\alpha}.\tag{6}
\end{align*}
Thus $\alpha \geq \min\{a,\,\beta\}$, where $\beta:=\max\{k:(x-\delta_{0,0})^k\mid ((x-\delta_{0,0})^{t_0}h_0(x) -h_0(x)^{-1}(x-\delta_{0,0})^{2a-p^s-t_0}-h_1(x)h_0(x)^{-1}(x-\delta_{0,0})^{a+t_1-t_0}h_1(x)h_0(x)^{-1})\}.$ In particular, since $u^2(x-\delta_{0,0})^{L}\in C$, we have $L\geq\min\{a,\,\beta\}.$ Also if we take $c_1(x)=1, c_2(x)=0$, we get $u^2(x-\delta_{0,0})^a \in C$ and if we take $c_1(x)=0, c_2(x)=1$, we get $u^2(x-\delta_{0,0})^{\beta} \in C$. Since $L$ is the smallest non-negative integer satisfying $u^2(x-\delta_{0,0})^{L} \in C$, we have $L\leq \min\{a,\beta\}$ and hence $L= \min\{a, \beta\}.$\hfill $\square$


\end{proof}
\begin{lemma}\label{third parameter computation}
        Let $L$ be the smallest non-negative integer such that $u^2(x-\delta_{0,0})^L\in \langle (x-\delta_{0,0})^{a}+u(x-\delta_{0,0})^{t_0}h_0(x)+u^2(x-\delta_{0,0})^{t_1}h_1(x),\, u(x-\delta_{0,0})^b+u^2(x-\delta_{0,0})^{t_2}h_2(x)\rangle$, where for 
    $i=0,1,$ and $2$, $h_i(x)$, if non-zero, is a unit in $\frac{\mathbb{F}_{p^m}[x]}{\langle (x-\delta_{0,0})^{p^s}\rangle}$. Then
        $$L= \begin{cases} 
             0 & \textnormal{ if } h_0(x)=0, \\
             0  & \textnormal{ if } h_0(x) \ne 0 \textnormal{ and } a<p^s-a+t_0, \\
              \Min\{b,\,\beta_1,\,\beta_2\}  & \textnormal{ if } h_0(x) \ne 0 \textnormal{ and } b\geq p^s-a+t_0, \\
             
   \end{cases}$$
   where $\beta_1:=\max\{k:(x-\delta_{0,0})^k\mid ((x-\delta_{0,0})^{t_0}h_0(x) -(x-\delta_{0,0})^{2a-p^s-t_0}h_0(x)^{-1}-(x-\delta_{0,0})^{a+t_1-t_0}h_1(x)h_0(x)^{-1})\}$ and $\beta_2:=\max\{k: (x-\delta_{0,0})^k\mid ((x-\delta_{0,0})^{t_2}h_2(x)-(x-\delta_{0,0})^{a+b-p^s-t_0}\\h_0(x)^{-1}-(x-\delta_{0,0})^{b+t_1-t_0}h_1(x)h_0(x)^{-1})\}.$
\end{lemma}
\begin{proof}
Let $C=\langle (x-\delta_{0,0})^{a}+u(x-\delta_{0,0})^{t_0}h_0(x)+u^2(x-\delta_{0,0})^{t_1}h_1(x),\, u(x-\delta_{0,0})^b+u^2(x-\delta_{0,0})^{t_2}h_2(x)\rangle$ and let $\alpha$ be a non-negative integer such that $u^2(x-\delta_{0,0})^{\alpha} \in C$. Then there exists  
$c_i(x)=\underset{j=0}{\overset{p^s-1}{\sum}}a_{i,j}(x-\delta_{0,0})^j,\,d_k(x)=\underset{j=0}{\overset{p^s-1}{\sum}}b_{k,j}(x-\delta_{0,0})^j\in\frac{\mathbb{F}_{p^m}[x]}{\langle (x-\delta_{0,0})^{p^s}\rangle}$ for $i=0,1,2$ and $k=0,1$ such that
\begin{align*}
  u^2(x-\delta_{0,0})^{\alpha}&=(x-\delta_{0,0})^a c_0(x) +u\bigl(( x^2-\delta_{0,0})^{b}d_0(x)+(x-\delta_{0,0})^ac_1(x)+ (x-\delta_{0,0})^{t_0}h_0(x)\\&c_0(x)\bigr)+
  u^2\bigl((x-\delta_{0,0})^ac_2(x)+(x-\delta_{0,0})^{t_0}h_0(x)c_1(x) +(x-\delta_{0,0})^{t_1}h_1(x)c_0(x)\\&+(x-\delta_{0,0})^{t_2}h_2(x)d_0(x)+(x-\delta_{0,0})^{b}d_1(x)\bigr).
  \end{align*}
    Let $\underset{j=0}{\overset{p^s-1-a}{\sum}}a_{0,j}(x-\delta_{0,0})^j=c_{0,0}(x)$ and $\underset{j=0}{\overset{a-1}{\sum}}a_{0,{j+p^s-a}}(x-\delta_{0,0})^j=c_{0,1}(x)$. Then $c_0(x)=c_{0,0}(x)+(x-\delta_{0,0})^{p^s-a}c_{0,1}(x).$  Since $\langle(x-\delta_{0,0})^{p^s}\rangle=\langle u^2\rangle$, we have
  \begin{align*}
   u^2(x-\delta_{0,0})^{\alpha}&=(x-\delta_{0,0})^a c_{0,0}(x) +u\bigl(( x^2-\delta_{0,0})^{b}d_0(x)+(x-\delta_{0,0})^ac_1(x)+ (x-\delta_{0,0})^{t_0}h_0(x)\\&(c_{0,0}(x)+(x-\delta_{0,0})^{p^s-a}c_{0,1}(x))\bigr)+
  u^2\bigl((x-\delta_{0,0})^ac_2(x)+c_{0,1}(x)+(x-\delta_{0,0})^{t_0}\\&h_0(x)c_1(x) +(x-\delta_{0,0})^{t_1}h_1(x)(c_{0,0}(x)+(x-\delta_{0,0})^{p^s-a}c_{0,1}(x))+(x-\delta_{0,0})^{t_2}\\&h_2(x)d_0(x)+(x-\delta_{0,0})^{b}d_1(x)\bigr).
\end{align*}
Hence, 
\begin{align*}
&(x-\delta_{0,0})^a c_{0,0}(x) = 0, \tag{1}\\
&( x^2-\delta_{0,0})^{b}d_0(x)+(x-\delta_{0,0})^ac_1(x)+ (x-\delta_{0,0})^{t_0}h_0(x)(c_{0,0}(x)+(x-\delta_{0,0})^{p^s-a}c_{0,1}(x))= 0, \tag{2}\\
&(x-\delta_{0,0})^ac_2(x)+c_{0,1}(x)+(x-\delta_{0,0})^{t_0}h_0(x)c_1(x) +(x-\delta_{0,0})^{t_1}h_1(x)(c_{0,0}(x)+(x-\delta_{0,0})^{p^s-a}\\&c_{0,1}(x))+(x-\delta_{0,0})^{t_2}h_2(x)d_0(x)+(x-\delta_{0,0})^{b}d_1(x)=(x-\delta_{0,0})^{\alpha}. \tag{3}
\end{align*}

Equation (1) gives $c_{0,0}(x) = 0$. 
Using this in Equations (2) and (3), we get
\begin{align*}
&( x^2-\delta_{0,0})^{b}d_0(x)+(x-\delta_{0,0})^ac_1(x)+ (x-\delta_{0,0})^{p^s-a+t_0}h_0(x)c_{0,1}(x) = 0, \tag{4}\\
&(x-\delta_{0,0})^ac_2(x)+c_{0,1}(x)+(x-\delta_{0,0})^{t_0}h_0(x)c_1(x) +(x-\delta_{0,0})^{p^s-a+t_1}h_1(x)c_{0,1}(x)+\\&(x-\delta_{0,0})^{t_2}h_2(x)d_0(x)+(x-\delta_{0,0})^{b}d_1(x)=(x-\delta_{0,0})^{\alpha}. \tag{5}
\end{align*}

If $h_0(x)= 0$ then $((x-\delta_{0,0})^{a}+u^2(x-\delta_{0,0})^{t_1}h_1(x))(x-\delta_{0,0})^{p^s-a}\in C$, which gives $u^2(1+(x-\delta_{0,0})^{p^s-a+t_1}h_1(x))\in C$. Since $p^s-a+t_1 \geq 1,$ we have $u^2\in C.$ So, $L=0.$\\
So, let $h_0(x) \neq 0$. We consider two cases.\\
\textbf{Case 1.}  $b < p^s-a+t_0$.\\ In this case, we have
\begin{align*}
    ((x-\delta_{0,0})^{a}+&u^2(x-\delta_{0,0})^{t_1}h_1(x))(x-\delta_{0,0})^{p^s-a}-\\&(u(x-\delta_{0,0})^b+u^2(x-\delta_{0,0})^{t_2}h_2(x))(x-\delta_{0,0})^{p^s-a-b+t_0}\in C,
\end{align*}
that is, $$u^2\{1+(x-\delta_{0,0})^{p^s-a+t_1}h_1(x)-(x-\delta_{0,0})^{p^s-a-b+t_0+t_2}h_0(x)h_2(x)\}\in C.$$
Since $p^s-a-b+t_0>0$ and $p^s-a+t_1>0,$ we get $u^2\in C.$ Hence, $L=0.$\\
\textbf{Case 2.}  $b \geq p^s-a+t_0$.\\
In this case, Equation (4) gives
\[{c}_{0,1}(x)=-h_0(x)^{-1}\{(x-\delta_{0,0})^{a+b-p^s-t_0}d_0(x)+(x-\delta_{0,0})^{2a-p^s-t_0}c_1(x)\}.\]Using this in Equation (5), we get 
\begin{align*}
&(x-\delta_{0,0})^ac_2(x)+\{(x-\delta_{0,0})^{t_0}h_0(x) -(x-\delta_{0,0})^{2a-p^s-t_0}h_0(x)^{-1}-(x-\delta_{0,0})^{a+t_1-t_0}\\&h_1(x)h_0(x)^{-1}\}c_1(x)+\{(x-\delta_{0,0})^{t_2}h_2(x)-(x-\delta_{0,0})^{a+b-p^s-t_0}h_0(x)^{-1}-(x-\delta_{0,0})^{b+t_1-t_0}\\&h_1(x)h_0(x)^{-1}\}d_0(x)+(x-\delta_{0,0})^{b}d_1(x)=(x-\delta_{0,0})^{\alpha}.\tag{6}
\end{align*}
Thus $\alpha \geq \min\{a,\,b,\,\beta_1,\,\beta_2\}=\min\{b,\,\beta_1,\,\beta_2\}$, where $\beta_1:=\max\{k:(x-\delta_{0,0})^k\mid ((x-\delta_{0,0})^{t_0}h_0(x) -(x-\delta_{0,0})^{2a-p^s-t_0}h_0(x)^{-1}-(x-\delta_{0,0})^{a+t_1-t_0}h_1(x)h_0(x)^{-1})\}$ and $\beta_2:=\max\{k: (x-\delta_{0,0})^k\mid ((x-\delta_{0,0})^{t_2}h_2(x)-(x-\delta_{0,0})^{a+b-p^s-t_0}h_0(x)^{-1}-(x-\delta_{0,0})^{b+t_1-t_0}h_1(x)h_0(x)^{-1})\}.$ In particular, since $u^2(x-\delta_{0,0})^{L}\in C$, we have $L\geq\min\{b,\,\beta_1,\,\beta_2\}.$ Also, if we take $d_1(x)=1, c_2(x)=c_1(x)=d_0(x)=0,$, we get $u^2(x-\delta_{0,0})^b \in C$ and if we take $c_1(x)=1, c_2(x)=d_1(x)=d_0(x)=0,$, we get $u^2(x-\delta_{0,0})^{\beta_1} \in C$ and if we take $d_1(x)=1, c_2(x)=c_1(x)=d_0(x)=0$, we get $u^2(x-\delta_{0,0})^b \in C$ and if we take $d_0(x)=1, c_2(x)=c_1(x)=d_0(x)=0$, we get $u^2(x-\delta_{0,0})^{\beta_2} \in C$. Since $L$ is the smallest non-negative integer satisfying $u^2(x-\delta_{0,0})^{L} \in C$, we have $L\leq \min\{b,\beta_1,\,\beta_2\}$ and hence $L= \min\{b, \beta_1,\, \beta_2\}.$\hfill $\square$
\end{proof}\begin{remark}
    Note that when computing the parameters, one may be tempted to write the equations by making a direct comparison of $u^2(x-\delta_{0,0})^{\alpha}$ with its original expression in terms of the generators of the ideal. This, however, may be misleading. The presence of the relation $\langle(x-\delta_{0,0})^{p^s}\rangle=\langle u^2\rangle$ makes computations challenging. One must, therefore, be cautious and use this relation before writing the equations.
\end{remark}
Before we give the cardinality of the ideals of $R^{3,1}_\delta$, when  $\delta_1= 0$ and $\delta_2\neq 0$, we recall that the $i$-th $ (0\leq i \leq 2)$ torsion of $C$, denoted by $\Tor_i(C)$, is the ideal
$$\Tor_i(C)=\mu(\{c(x)\in \frac{\mathbb{F}_{p^m}[u]}{\langle u^3\rangle}[x]/\langle x^{p^s}-\delta_0-u^2\delta_2\rangle: c(x)u^i\in C\}).$$
$\frac{\mathbb{F}_{p^m}[x]}{\langle (x-\delta_{0,0})^{p^s}\rangle}$. $\Tor_0(C)$ is referred as residue of $C,$ and is denoted by $\textnormal{Res}(C).$ In fact, $\Tor_i(C)=\langle (x-\delta_{0,0})^{T_i}\rangle $ for some integer $T_i$ such that $0\leq T_i\leq p^s.$ $T_i$ is called the $i$-th torsional degree of $C$. In the following lemma, we give the $i$-th torsional degree for the ideals given in Theorem \ref{ideals of ps specific case}.
\begin{lemma}\label{Torsions}
    Let $C$ be an ideal of $\frac{R^3[x]}{\langle x^{p^s}-\delta_0 -u^2\delta_2\rangle}$ given in Theorem \ref{ideals of ps specific case}. Then
    \begin{enumerate}
        \item If $C=\langle 0 \rangle$, then $T_0(C)=T_1(C)=T_2(C)=p^s.$
        \item If $C=\langle 1 \rangle,$ then $T_0(C)=T_1(C)=T_2(C)=0.$
        \item If $C$ is described in Theorem \ref{ideals of ps specific case} (\ref{Type 2'}), then $T_0(C)=T_1(C)=p^s,$ and $T_2(C)=a.$
        \item If $C$ is described in Theorem \ref{ideals of ps specific case} (\ref{Type 3'}), then $T_0(C)=p^s,\, T_1(C)=a$, and $T_2(C)=L,$ where $L$ is as mentioned in Theorem \ref{ideals of ps specific case} (\ref{Type 3'}).
        \item If $C$ is described in Theorem \ref{ideals of ps specific case} (\ref{Type 4'}), then $T_0(C)=p^s,\, T_1(C)=a$, and $T_2(C)=b.$
         \item If $C$ is described in Theorem \ref{ideals of ps specific case} (\ref{Type 5'}), then $T_0(C)=a,\, T_1(C)=L$, and $T_2(C)=M,$ where $L$ and $M$ are as mentioned in Theorem \ref{ideals of ps specific case} (\ref{Type 5'}).
         \item If $C$ is described in Theorem \ref{ideals of ps specific case} (\ref{Type 6'}), then $T_0(C)=a,\, T_1(C)=M$, and $T_2(C)=b,$ where $M$ is as mentioned in Theorem \ref{ideals of ps specific case} (\ref{Type 6'}). 
         \item If $C$ is described in Theorem \ref{ideals of ps specific case} (\ref{Type 7'}), then $T_0(C)=a,\, T_1(C)=b$, and $T_2(C)=M,$ where $M$ is as mentioned in Theorem \ref{ideals of ps specific case} (\ref{Type 7'}).
         \item If $C$ is described in Theorem \ref{ideals of ps specific case} (\ref{Type 8'}), then $T_0(C)=a,\, T_1(C)=b$, and $T_2(C)=c$.
    \end{enumerate}  
    \end{lemma}
    \begin{proof}
        The proof is similar to that of Lemma 4.7 in \cite{AkaKanSar}. For the sake of completeness, we outline it here.\\
        If $C=\langle 0\rangle,$ then clearly $T_0(C)=T_1(C)=T_2(C)=p^s.$ If $C=\langle 1\rangle,$ then clearly $T_0(C)=T_1(C)=T_2(C)=0.$ If $C=\langle  u^2(x-\delta_{0,0})^{a}\rangle,$ then $T_0(C)=T_1(C)=p^s$. By definition, $\Tor_2(C)=\mu\{c(x)\in R^{3,1}_{\delta} \,|\, c(x)u^2\in C\}.$ Note that $\mu((x-\delta_{0,0})^{a})\in \Tor_2(C)\textnormal{ and hence } \langle (x-\delta_{0,0})^{a}\subset \Tor_2(C).$ Conversely, if $\mu(a(x))\in \Tor_2(C),$ for some $a(x)\in R^{3,1}_{\delta},$ then $a(x)u^2\in C\implies a(x)u^2=u^2(x-\delta_{0,0})^{a}h(x)$ for some $h(x)\in R^{3,1}_{\delta}.$ Then we have $u^2\underset{j=0}{\overset{p^s-1}{\sum}}a_{j}^{(0)}(x-\delta_{0,0})^j=u^2(x-\delta_{0,0})^{a}\underset{j=0}{\overset{p^s-1}{\sum}}h_{j}^{(0)}f(x)^j.$ Hence $\mu(a(x))\in \langle (x-\delta_{0,0})^{a}\rangle\implies \Tor_2(C)\subset \langle (x-\delta_{0,0})^{a}\rangle.$ 
        The procedure for calculating torsional degrees in other cases is similar.\hfill{$\square$}
    \end{proof}
    Using Lemma 4.5 and Theorem 4.6 in \cite{AkaKanSar} and Lemma \ref{Torsions} above, we have the following theorem.
    \begin{theorem}\label{Cardinality}
     Let $C$ be an ideal of $\frac{R^3[x]}{\langle x^{p^s}-\delta_0 -u^2\delta_2\rangle}$ given in Theorem \ref{ideals of ps specific case}. Then
            \begin{enumerate}
        \item If $C=\langle 0 \rangle$, then $|C|=1.$
        \item If $C=\langle 1 \rangle,$ then $|C|=p^{3mp^s}.$
        \item If $C$ is described in Theorem \ref{ideals of ps specific case} (\ref{Type 2'}), then $|C|=p^{m(p^s-a)}.$
        \item If $C$ is described in Theorem \ref{ideals of ps specific case} (\ref{Type 3'}), then $|C|=p^{m(2p^s-a-L)},$ where $L$ is as mentioned in Theorem \ref{ideals of ps specific case} (\ref{Type 3'}).
        \item If $C$ is described in Theorem \ref{ideals of ps specific case} (\ref{Type 4'}), then $|C|=p^{m(2p^s-a-b)}.$
         \item If $C$ is described in Theorem \ref{ideals of ps specific case} (\ref{Type 5'}), then $|C|=p^{m(3p^s-a-L-M)},$ where $L$ and $M$ are as mentioned in Theorem \ref{ideals of ps specific case} (\ref{Type 5'}).
         \item If $C$ is described in Theorem \ref{ideals of ps specific case} (\ref{Type 6'}), then $|C|=p^{m(3p^s-a-b-M)},$ where $M$ is as mentioned in Theorem \ref{ideals of ps specific case} (\ref{Type 6'}).
         \item If $C$ is described in Theorem \ref{ideals of ps specific case} (\ref{Type 7'}), then $|C|=p^{m(3p^s-a-b-M)},$ where $M$ is as mentioned in Theorem \ref{ideals of ps specific case} (\ref{Type 7'}).
         \item If $C$ is described in Theorem \ref{ideals of ps specific case} (\ref{Type 8'}), then $|C|=p^{m(3p^s-a-b-c)}.$
    \end{enumerate}
    \end{theorem}
\section{Ideals of \texorpdfstring{$R^{t,2}_{\delta}$}{}}\label{section4}
In this section, we use the structure of ideals of $R^{t,n}_{\delta}:=\frac{R^t [x]}{\langle x^{np^s}-\delta\rangle},$ where $\delta = \delta_0+u\delta_1+\dots+u^{t-1}\delta_{t-1}$ is a unit in $R^t:=\frac{\mathbb{F}_{p^m}[u]}{\langle u^t \rangle}$, $\delta_i\in \mathbb{F}_{p^m}$ for $0\leq i \leq t-1$ given in Section \ref{section2} to discuss the case when $n=2$ and $p$ is an odd prime, that is, to discuss the ideals of $R^{t,2}_{\delta}:=\frac{R^t [x]}{\langle x^{2p^s}-\delta\rangle}$. We note that $\delta_0\ne 0$ since $\delta$ is a unit.
 
 We first assume that $\delta$ is a square in $R^t$, that is, there exists $\tilde{\delta}=\tilde{\delta}_0+u\tilde{\delta}_1+\dots+u^{t-1}\tilde{\delta}_{t-1}\in R^t$ such that $\tilde{\delta}^2=\delta$. Then,
$$x^{2p^s}-\delta=(x^{p^s}-\tilde{\delta})(x^{p^s}+\tilde{\delta}).$$
Also, since $\delta$ is a unit in $R^t$, $\tilde{\delta}$ is a unit in $R^t$. Moreover, as $\tilde{\delta}\ne -\tilde{\delta},$ we have
$$R^{t,2}_{\delta}\cong \frac{R^t[x]}{\langle x^{p^s}-\tilde{\delta}\rangle} \times\frac{R^t[x]}{\langle x^{p^s}+\tilde{\delta} \rangle}.$$
Thus, the ideals of the ring $R^{t,2}_{\delta}$ are of the form $I_1 \times I_2$, where $I_1$ is an ideal of $R^{t,1}_{\tilde\delta}$ and $I_2$ is an ideal of $R^{t,1}_{-\tilde\delta}$. We can, now, as in Section \ref{section3}, get the structure of the ideal $I_1$ by taking $n=1$ and replacing $x-\delta_{0,0}$ with $x-\tilde{\delta}_{0,0}$ in Theorem \ref{nps ideals} and that of $I_2$ by taking $n=1$ and replacing $x-\delta_{0,0}$ with $x+\tilde{\delta}_{0,0}$ in Theorem \ref{nps ideals}.

Next, we assume that $\delta$ is not a square. By Proposition \ref{delta power n iff delta_0 power n}, $\delta_{0}$ is not a square. Also, using the notation in Section \ref{section2}, if $\delta_{0,0} \in \mathbb{F}_{p^m}$ is such that $\delta_{0,0}^{p^s}=\delta_0$ then $\delta_{0,0}$ is also not a square. Thus, $x^2-\delta_{0,0} \in \mathbb{F}_{p^m}[x]$ is an irreducible polynomial. 
Hence, taking $n=2$ in Theorem \ref{nps ideals}, we get the generators and the structure of ideals of $R^{t,2}_{\delta}$, in this case.\\
Note that, as in Section \ref{section2} and Section \ref{section3}, if $k$ is the smallest integer such that $\delta_k \neq0$, where $1\leq k\leq t-1$, then $\langle (x^2-\delta_{0,0})^{p^s} \rangle=\langle u^k\rangle$, $x^2-\delta_{0,0}$ is nilpotent in $R^{t,2}_{\delta}$ with nilpotency index $\lceil \frac{t}{k}\rceil p^s$, and an arbitrary element $c(x)$ of $R^{t,2}_{\delta}$ can be uniquely written as
        \begin{align*}
            c&(x)=\underset{j=0}{\overset{p^s-1}{\sum}}(a_{j}^{0}+b_{j}^{0}x)(x^2-\delta_{0,0})^j+\dots+u^{t-1}\underset{j=0}{\overset{p^s-1}{\sum}}(a_{j}^{t-1}+b_{j}^{t-1}x)(x^2-\delta_{0,0})^j,
        \end{align*}
        where $a_{j}^k,b_{j}^k\in \mathbb{F}_{p^m}$ for $0\leq j\leq p^s-1$ and $0\leq k\leq t-1.$ Further, an element $c(x)$, as given above, is a non-unit if and only if $a_{0}^0=0$ and $b_{0}^0=0$.
        
Also, as in  Corollary \ref{chainfornpscase}, the ideals of $R^{t,2}_{\delta}$ form a chain if and only if $\delta_1\ne 0$. In fact, the ideals  of $R^{t,2}_{\delta}$, in this case, are
        $$\langle 0\rangle\subset \langle (x^2-\delta_{0,0})^{tp^s-1}\rangle \subset \langle (x^2-\delta_{0,0})^{tp^s-2}\rangle \subset \dots \subset \langle x^2-\delta_{0,0}\rangle \subset \langle 1\rangle.$$ 
Moreover, for $0\leq i\leq tp^s$, the number of codewords in $\langle (x^2-\delta_{0,0})^i\rangle$ is $p^{2m(tp^s-i)}$.

In the next theorem, we summarize the description of ideals of the ring $R^{t,2}_{\delta}$ in all cases.
\begin{theorem}
    \label{final theorem for section 4}Let $\delta=\delta_0+u\delta_1+\dots+u^{t-1}\delta_{t-1} \in R^t$ be a unit and let $\delta_{0, 0}^{p^s}=\delta_0$.
\begin{enumerate}
    \item[(a)] If $\delta$ is a square in $R^t$, say $\delta=\tilde{\delta}^2$, then $R^{t, 2}_\delta\cong R^{t, 1}_{\tilde\delta}\times R^{t, 1}_{-\tilde\delta}$ and any ideal of $R^{t, 2}_\delta$ has the form $I_1 \times I_2$ where $I_1$ and $I_2$ are ideals of $R^{t, 1}_{\tilde{\delta}}$ and $R^{t, 1}_{-\tilde\delta}$, respectively. The structure of $I_1$ and $I_2$ can be obtained by using Theorem \ref{nps ideals} with appropriate substitutions, as described above.
\item[(b)]  If $\delta$ is not a square then the ideals of the ring $R^{t,2}_{\delta}$ are given by Theorem \ref{nps ideals} for $n=2$.
\end{enumerate}
\end{theorem}
Now, as a special case, we give all the types of ideals of $R^{3,2}_\delta$ and give their cardinalities, that is, we discuss the case when $t=3$. Let $\delta=\delta_0 + u\delta_1 + u^2\delta_2 \in R^3$ be a unit. If $\delta$ is a square, say $\delta=\tilde{\delta}^2$, then as discussed at the beginning of this section, any ideal $I$ of $\frac{R^t[x]}{\langle x^{2p^s}-\delta\rangle}$ has the form $I_1 \times I_2 $ where $I_1$ is an ideal of $\frac{R^t[x]}{\langle x^{p^s}+\tilde{\delta}\rangle}$ and $I_2$ is an ideal of $\frac{R^t[x]}{\langle x^{p^s}-\tilde{\delta}\rangle}$. We can now use Lemma \ref{Torsions} and Theorem \ref{Cardinality} to get the torsional degrees and cardinalities of $I_1$ and $I_2$. The cardinality of $I$ is, then, the product of the cardinalities of $I_1$ and $I_2$. 
If $\delta$ is not a square and $\delta_1 \ne 0$ or $\delta_1=0,\,\delta_2= 0$, we refer the reader to \cite{Sriwirach}. 
Therefore, let $\delta$ be not a square and $\delta_1=0,\,\delta_2\neq 0$, that is, $\delta=\delta_0 + u^2\delta_2$.   
    With these notations, we, in the next theorem, list all the different types of ideals of $R^{3,2}_{\delta}$, in this case. 
    \begin{theorem}\label{ideals whose parameters are to be computed}
       Suppose that $\delta=\delta_0+u^{2}\delta_{2}\in R^3$ be a unit which is not a square and let $\delta_{0, 0}^{p^s}=\delta_0$. Then the ideals of the ring $R^{3,2}_{\delta}$ have one of the following eight types:
       \begin{enumerate}
          \item{\label{Type 1}} $\langle 0 \rangle,\, \langle 1 \rangle.$
           \item{\label{Type 2}} $\langle u^2 (x^2-\delta_{0,0})^a \rangle$ where $0\leq a \leq p^s-1.$
           \item{\label{Type 3}} $\langle u(x^2-\delta_{0,0})^a+u^2(x^2-\delta_{0,0})^{t}h(x)\rangle$ where $0\leq t<L< a\leq p^s-1,$ $h(x)$ is either $0$ or a unit in $\frac{\mathbb{F}_{p^m}[x]}{\langle (x^2-\delta_{0,0})^{p^s}\rangle},$ and $L$ is the smallest non-negative integer such that $u^2(x^2-\delta_{0,0})^L\in \langle u(x^2-\delta_{0,0})^a+u^2(x^2-\delta_{0,0})^{t}h(x)\rangle.$
           
           \item{\label{Type 4}}$\langle u(x^2-\delta_{0,0})^a+u^2(x^2-\delta_{0,0})^{t}h(x),\, u^2(x^2-\delta_{0,0})^b\rangle$ where $0\leq t<b<L< a\leq p^s-1$, $h(x)$ is either $0$ or a unit in $\frac{\mathbb{F}_{p^m}[x]}{\langle (x^2-\delta_{0,0})^{p^s}\rangle},$ and $L$ the smallest non-negative integer such that $u^2(x^2-\delta_{0,0})^{L}\in \langle u(x^2-\delta_{0,0})^a+u^2(x^2-\delta_{0,0})^{t}h(x)\rangle.$
           
           \item{\label{Type 5}} $\langle (x^2-\delta_{0,0})^{a}+u(x^2-\delta_{0,0})^{t_0}h_0(x)+u^2(x^2-\delta_{0,0})^{t_1}h_1(x)\rangle$ where $0\leq t_1<t_0<L< a\leq p^s-1,\,0\leq t_1<M<L,\,$ $h_i(x)$ is either $0$ or a unit in $\frac{\mathbb{F}_{p^m}[x]}{\langle (x^2-\delta_{0,0})^{p^s}\rangle}$ for $i=0,\,1$ and $L,M$ are the smallest non-negative integers such that $u(x^2-\delta_{0,0})^L+u^2g(x),\, u^2(x^2-\delta_{0,0})^M\in \langle (x^2-\delta_{0,0})^{a}+u(x^2-\delta_{0,0})^{t_0}h_0(x)+u^2(x^2-\delta_{0,0})^{t_1}h_1(x)\rangle$ for some $g(x)\in \frac{\mathbb{F}_{p^m}[x]}{\langle (x^2-\delta_{0,0})^{p^s}\rangle}.$

           \item{\label{Type 6}}$\langle (x^2-\delta_{0,0})^{a}+u(x^2-\delta_{0,0})^{t_0}h_0(x)+u^2(x^2-\delta_{0,0})^{t_1}h_1(x),\,u^2 (x^2-\delta_{0,0})^b \rangle\rangle$ for $0\leq t_1<t_0< a\leq p^s-1,\, 0\leq t_1<b <L<a,\,0\leq t_0<M<L$, $h_i(x)$ is either $0$ or a unit in $\frac{\mathbb{F}_{p^m}[x]}{\langle (x^2-\delta_{0,0})^{p^s}\rangle}$ for $i=0,\,1,$ and $L,M$ are the smallest non-negative integers such that $u^2(x^2-\delta_{0,0})^L,\, u(x^2-\delta_{0,0})^M+u^2g(x)\in \langle (x^2-\delta_{0,0})^{a}+u(x^2-\delta_{0,0})^{t_0}h_0(x)+u^2(x^2-\delta_{0,0})^{t_1}h_1(x) \rangle$ for some $g(x)\in \frac{\mathbb{F}_{p^m}[x]}{\langle (x^2-\delta_{0,0})^{p^s}\rangle}.$

           \item{\label{Type 7}} $\langle (x^2-\delta_{0,0})^{a}+u(x^2-\delta_{0,0})^{t_0}h_0(x)+u^2(x^2-\delta_{0,0})^{t_1}h_1(x),\, u(x^2-\delta_{0,0})^b+u^2(x^2-\delta_{0,0})^{t_2}h_2(x)\rangle$ where $0\leq t_1<t_0<b<L< a\leq p^s-1\,, 0\leq t_2<b,\,0\leq t_i<M<b$ for $i=1,2$, $h_i(x)$ is either $0$ or a unit in $\frac{\mathbb{F}_{p^m}[x]}{\langle (x^2-\delta_{0,0})^{p^s}\rangle}$ for $0\leq i\leq 2,$ $L,M$ are the smallest non-negative integers such that $u(x^2-\delta_{0,0})^L+u^2g(x)\in \langle (x^2-\delta_{0,0})^{a}+u(x^2-\delta_{0,0})^{t_0}h_0(x)+u^2(x^2-\delta_{0,0})^{t_1}h_1(x)\rangle$ and $u^2(x^2-\delta_{0,0})^M\in \langle (x^2-\delta_{0,0})^{a}+u(x^2-\delta_{0,0})^{t_0}h_0(x)+u^2(x^2-\delta_{0,0})^{t_1}h_1(x),\, u(x^2-\delta_{0,0})^b+u^2(x^2-\delta_{0,0})^{t_2}h_2(x)\rangle$ for some $g(x)\in \frac{\mathbb{F}_{p^m}[x]}{\langle (x^2-\delta_{0,0})^{p^s}\rangle}.$

           \item{\label{Type 8}}$\langle (x^2-\delta_{0,0})^{a}+u(x^2-\delta_{0,0})^{t_0}h_0(x)+u^2(x^2-\delta_{0,0})^{t_1}h_1(x),\, u(x^2-\delta_{0,0})^b+u^2(x^2-\delta_{0,0})^{t_2}h_2(x),\, \\u^2(x^2-\delta_{0,0})^c\rangle$ where $0\leq t_1<t_0< a\leq p^s-1,\, 0\leq t_1,t_2<c<b<a,\,0\leq t_0<b$,  $h_i(x)$ is either $0$ or a unit in $\frac{\mathbb{F}_{p^m}[x]}{\langle (x^2-\delta_{0,0})^{p^s}\rangle}$ for $0\leq i\leq 2$, $c<M$, the smallest non-negative integer such that $u^2(x^2-\delta_{0,0})^M\in \langle (x^2-\delta_{0,0})^{a}+u(x^2-\delta_{0,0})^{t_0}h_0(x)+u^2(x^2-\delta_{0,0})^{t_1}h_1(x),\, u(x^2-\delta_{0,0})^b+u^2(x^2-\delta_{0,0})^{t_2}h_2(x)\rangle$, and $b<L$, the smallest non-negative integer such that $u(x^2-\delta_{0,0})^L+u^2g(x)\in \langle (x^2-\delta_{0,0})^{a}+u(x^2-\delta_{0,0})^{t_0}h_0(x)+u^2(x^2-\delta_{0,0})^{t_1}h_1(x)\rangle$ for some $g(x)\in \frac{\mathbb{F}_{p^m}[x]}{\langle (x^2-\delta_{0,0})^{p^s}\rangle}.$
       \end{enumerate}
    \end{theorem}
 As in Section \ref{section3}, we can give the torsional degree as well as cardinalities of the ideals in Theorem \ref{ideals whose parameters are to be computed}. The results are similar to Lemma \ref{Torsions} and Theorem \ref{Cardinality}. The critical piece for obtaining the torsional degree and cardinality is once again obtaining the parameters used in the description of ideals in Theorem \ref{ideals whose parameters are to be computed}. For the sake of completeness, we state here the results giving the parameters. The proofs of these results and computations are similar to those of Lemmas \ref{first parameter computation}, Lemma \ref{second parameter computation}, and Lemma \ref{third parameter computation}, respectively. The main differences include using $c_i(x)=\underset{j=0}{\overset{p^s-1}{\sum}}(a_{i,j}+b_{i,j}x)(x^2-\delta_{0,0})^j$, where $a_{i,j},b_{i,j}\in \mathbb{F}_{p^m}$ for $i=0,1,\,2$ and $0\leq j \leq p^{p^s-1}$, writing $c_0(x)=c_{0,0}(x)+(x^2-\delta_{0,0})^{p^s-a}c_{0,1}(x)$ where $c_{0,0}(x)=\underset{j=0}{\overset{p^s-1-a}{\sum}}(a_{0,j}+b_{0,j}x)(x^2-\delta_{0,0})^j$ and $c_{0,1}(x)=\underset{j=0}{\overset{a-1}{\sum}}(a_{0,{j+p^s-a}}+b_{0,{j+p^s-a}}x)(x^2-\delta_{0,0})^j,$ and replacing $(x-\delta_{0,0})$ with $(x^2-\delta_{0,0}).$\\
    \begin{lemma}
        Let $L$ be the smallest non-negative integer such that  $u^2(x^2-\delta_{0,0})^L\in \langle u(x^2-\delta_{0,0})^a+u^2(x^2-\delta_{0,0})^{t}h(x)\rangle,$ where $h(x)$, if non-zero, is a unit in $\frac{\mathbb{F}_{p^m}[x]}{\langle (x^2-\delta_{0,0})^{p^s}\rangle}$. Then
        $$L= \begin{cases} 
             a & \textnormal{ if } h(x)=0, \\
             \Min\{a,\,p^s-a+t\}  & \textnormal{ if } h(x) \ne 0. \\
   \end{cases}$$
\end{lemma}
\begin{lemma}
        Let $L$ be the smallest non-negative integer such that $u^2(x^2-\delta_{0,0})^L\in \langle (x^2-\delta_{0,0})^{a}+u(x^2-\delta_{0,0})^{t_0}h_0(x)+u^2(x^2-\delta_{0,0})^{t_1}h_1(x)\rangle$,  where for $i=0$ and $1$, $h_i(x)$, if non-zero, is a unit in $\frac{\mathbb{F}_{p^m}[x]}{\langle (x^2-\delta_{0,0})^{p^s}\rangle}.$ Then
        $$L= \begin{cases} 
             0 & \textnormal{ if } h_0(x)=0, \\
             0&\textnormal{ if } h_0(x)\neq 0 \textnormal{ and } a<p^s-a+t_0,\\
             \Min\{a,\,\beta\}  & \textnormal{ if } h_0(x) \ne 0 \textnormal{ and } a\geq p^s-a+t_0, \\
   \end{cases}$$
   where $\beta:=\max\{k:(x^2-\delta_{0,0})^k\mid ((x^2-\delta_{0,0})^{t_0}h_0(x) -h_0(x)^{-1}(x^2-\delta_{0,0})^{2a-p^s-t_0}-h_1(x)h_0(x)^{-1}(x^2-\delta_{0,0})^{a+t_1-t_0}h_1(x)h_0(x)^{-1})\}.$
\end{lemma}
\begin{lemma}
        Let $L$ be the smallest non-negative integer such that $u^2(x^2-\delta_{0,0})^L\in \langle (x^2-\delta_{0,0})^{a}+u(x^2-\delta_{0,0})^{t_0}h_0(x)+u^2(x^2-\delta_{0,0})^{t_1}h_1(x),\, u(x^2-\delta_{0,0})^b+u^2(x^2-\delta_{0,0})^{t_2}h_2(x)\rangle$  where for $i=0,1,$ and $2$, $h_i(x)$, if non-zero, is a unit in $\frac{\mathbb{F}_{p^m}[x]}{\langle (x^2-\delta_{0,0})^{p^s}\rangle}$. Then
        $$L= \begin{cases} 
             0 & \textnormal{ if } h_0(x)=0, \\
             0  & \textnormal{ if } h_0(x) \ne 0 \textnormal{ and } a<p^s-a+t_0, \\
              \Min\{b,\,\beta_1,\,\beta_2\}  & \textnormal{ if } h_0(x) \ne 0 \textnormal{ and } b\geq p^s-a+t_0, \\
             
   \end{cases}$$
   where $\beta_1:=\max\{k:(x^2-\delta_{0,0})^k\mid ((x^2-\delta_{0,0})^{t_0}h_0(x) -(x^2-\delta_{0,0})^{2a-p^s-t_0}h_0(x)^{-1}-(x^2-\delta_{0,0})^{a+t_1-t_0}h_1(x)h_0(x)^{-1})\}$ and $\beta_2:=\max\{k: (x^2-\delta_{0,0})^k\mid ((x^2-\delta_{0,0})^{t_2}h_2(x)-(x^2-\delta_{0,0})^{a+b-p^s-t_0}h_0(x)^{-1}-(x^2-\delta_{0,0})^{b+t_1-t_0}h_1(x)h_0(x)^{-1})\}.$
\end{lemma}
\section{Ideals of \texorpdfstring{$R^{t,3}_{\delta}$}{}}\label{section5}
In this section, we use the structure of ideals of $R^{t,n}_{\delta}:=\frac{R^t [x]}{\langle x^{np^s}-\delta\rangle},$ where $\delta = \delta_0+u\delta_1+\dots+u^{t-1}\delta_{t-1}$ is a unit in $R^t:=\frac{\mathbb{F}_{p^m}[u]}{\langle u^t \rangle}$, $\delta_i\in \mathbb{F}_{p^m}$ for $0\leq i \leq t-1$ given in Section \ref{section2} to discuss the case when $n=3$ and $\gcd(3,p)=1$, that is, to discuss the ideals of $R^{t,3}_{\delta}:=\frac{R^t [x]}{\langle x^{3p^s}-\delta\rangle}$. We note that $\delta_0\ne 0$ since $\delta$ is a unit. Also, since $\gcd(3,p)=1$, $p^m\equiv 1 (\textnormal{mod } 3)$ or $p^m\equiv 2 (\textnormal{mod } 3)$.

We first assume that $\delta$ is a cube in $R^t$, that is, there exists $\tilde{\delta}=\tilde{\delta}_0+u\tilde{\delta}_1+\dots+u^{t-1}\tilde{\delta}_{t-1}\in R^t$ such that $\tilde{\delta}^3=\delta$. Also, since $\delta$ is a unit in $R^t$, $\tilde{\delta}$ is a unit in $R^t$.  Moreover, 
$$x^{3p^s}-\delta=(x^{p^s}-\tilde{\delta})(x^{2p^s}+\tilde{\delta}x^{p^s}+\tilde{\delta}^2).$$ 
Since these are coprime factors, we have 
\begin{equation}
R^{t,3}_{\delta}\cong \frac{R^t[x]}{\langle x^{p^s}-\tilde{\delta}\rangle} \times\frac{R^t[x]}{\langle x^{2p^s}+\tilde{\delta}x^{p^s}+\tilde{\delta}^2 \rangle}. \label{chinese1}    
\end{equation}
Since the structure of ideals of $R^{t,1}_{\tilde\delta}$ can be obtained by taking $n=1$ and replacing $x-\delta_{0,0}$ with $x-\tilde{\delta}_{0,0}$ in Theorem \ref{nps ideals}, it is enough to give the structure of ideals of $\frac{R^t[x]}{\langle x^{2p^s}+\tilde{\delta}x^{p^s}+\tilde{\delta}^2 \rangle}$.

\noindent \textbf{Case 1.} $p^m\equiv 1 (\textnormal{ mod } 3)$.\\
By Lemma 8 in \cite{dinh2020constacyclic}, there exist $b,c \in \mathbb{F}_{p^m}\setminus\{0, 1\}$ such that $b\ne c$ and  
$x^{2p^s}+\tilde{\delta}x^{p^s}+\tilde{\delta}^2=(x^{p^s}-b\tilde{\delta})(x^{p^s}-c\tilde{\delta})$. Hence,
$$\frac{R^t[x]}{\langle x^{3p^s}-\delta\rangle}\cong \frac{R^t[x]}{\langle x^{p^s}-\tilde{\delta}\rangle} \times\frac{R^t[x]}{\langle x^{p^s}-b\tilde{\delta} \rangle}\times \frac{R^t[x]}{\langle x^{p^s}-c\tilde{\delta} \rangle}.$$
Thus, the ideals of the ring $R^{t,3}_{\delta}$ are of the form $I_1 \times I_2 \times I_3$ where $I_1$ is an ideal of $R^{t,1}_{\tilde\delta}$, $I_2$ is an ideal of $R^{t,1}_{b\tilde\delta}$, and $I_3$ is an ideal of $R^{t,1}_{c\tilde\delta}$. We can, now, get the structure of the ideal $I_1$ by taking $n=1$ and replacing $x-\delta_{0,0}$ with $x-\tilde{\delta}_{0,0}$ in Theorem \ref{nps ideals}, that of $I_2$ by taking $n=1$ and replacing $x-\delta_{0,0}$ with $x-b\tilde{\delta}_{0,0}$ in Theorem \ref{nps ideals}, and that of $I_3$ by taking $n=1$ and replacing $x-\delta_{0,0}$ with $x-c\tilde{\delta}_{0,0}$ in Theorem \ref{nps ideals}.

\textbf{Case 2.} $p^m\equiv 2 (\textnormal{ mod } 3)$.\\
 Recall that $\tilde{\delta}=\tilde{\delta}_0+u\tilde{\delta}_1+\dots+u^{t-1}\tilde{\delta}_{t-1}$, where $\tilde{\delta}_i \in \mathbb{F}_{p^m}$ for $0\leq i\leq t-1$, is a unit in $ {R^t}$. Also for $0\leq i\leq t-1$, there exists $\tilde{\delta}_{i,0}$ in $\mathbb{F}_{p^m}$ such that $\tilde{\delta}_{i,0}^{p^s}=\tilde{\delta}_i$. Let $k$ be the smallest integer such that $1 \leq k \leq t-1$ and $\delta_k \neq 0$. We note that this is equivalent to $k$ being the smallest integer such that $\tilde{\delta}_k \neq 0$. Then, in $\frac{R^t[x]}{\langle x^{2p^s}+\tilde{\delta}x^{p^s}+\tilde{\delta}^2 \rangle}$
$$x^{2p^s}+\tilde{\delta}x^{p^s}+\tilde{\delta}^2=(x^2+\tilde{\delta}_{0,0}x+\tilde{\delta}_{0,0}^2)^{p^s}+u^{k}z+u^{k+1}(\textnormal{some element of }\frac{R^t[x]}{\langle x^{2p^s}+\tilde{\delta}x^{p^s}+\tilde{\delta}^2 \rangle}),$$
where $z=x^{p^s}\tilde{\delta}_k+2\tilde{\delta}_k\tilde{\delta}_0=\tilde{\delta}_k(x^{p^s}+2\tilde{\delta}_0)$. Thus, in $\frac{R^t[x]}{\langle x^{2p^s}+\tilde{\delta}x^{p^s}+\tilde{\delta}^2 \rangle}$,
$$(x^{2}+\tilde{\delta}_{0,0}x+\tilde{\delta}_{0,0}^2)^{p^s}=-u^k\left(\tilde{\delta}_k(x^{p^s}+2\tilde{\delta}_0)+u(\textnormal{some element of }\frac{R^t[x]}{\langle x^{2p^s}+\tilde{\delta}x^{p^s}+\tilde{\delta}^2 \rangle})\right).$$ 

Also, by Lemma 11 in \cite{dinh2020constacyclic}, $x^{2}+\tilde{\delta}_{0,0}x+\tilde{\delta}_{0,0}^2$ is irreducible over $\mathbb{F}_{p^m}$. Thus, by Corollary \ref{ax+b invertible in 2ps case}, any non-zero polynomial of the form $p(x)=ax+b\in\mathbb{F}_{p^m}[x]$ is invertible in $\frac{R^t[x]}{\langle x^{2p^s}+\tilde{\delta}x^{p^s}+\tilde{\delta}^2\rangle}$. Also, there exists $\alpha_0\in \mathbb{F}_{p^m}$ such that $\alpha_0^{p^s}=2$. Thus, we can write $z$ as $\{\tilde{\delta}_{k,0}(x+\alpha_0\tilde{\delta}_{0,0})\}^{p^s}$. Since every non-zero polynomial of the form $ax+b$ is invertible, $\tilde{\delta}_{k,0}(x+\alpha_0\tilde{\delta}_{0,0})$ and hence $\{\tilde{\delta}_{k,0}(x+\alpha_0\tilde{\delta}_{0,0})\}^{p^s}$ is invertible in $\frac{R^t[x]}{\langle x^{2p^s}+\tilde{\delta}x^{p^s}+\tilde{\delta}^2\rangle}$. 
Hence, in $\frac{R^t[x]}{\langle x^{2p^s}+\tilde{\delta}x^{p^s}+\tilde{\delta}^2\rangle},$ $\langle (x^{2}+\tilde{\delta}_{0,0}x+\tilde{\delta}_{0,0}^2)^{p^s}\rangle=\langle u^k\rangle$. Moreover, $x^{2}+\tilde{\delta}_{0,0}x+\tilde{\delta}_{0,0}^2$ is nilpotent with nilpotency index $\lceil \frac{t}{k}\rceil p^s.$
Using this and the irreducibility of $x^2+\tilde{\delta}_{0,0}x+\tilde{\delta}_{0,0}^2$ over $\mathbb{F}_{p^m}$, we see, by Remark \ref{generalrepresentation}, that an arbitrary element $c(x)$ of $\frac{R^t[x]}{\langle x^{2p^s}+\tilde{\delta}x^{p^s}+\tilde{\delta}^2\rangle}$ can be uniquely written as  
\begin{align*}
            c(x)=&\underset{i=0}{\overset{p^s-1}{\sum}}(a_{0,i}^{0}+a_{1,i}^{0}x)(x^2+\tilde{\delta}_{0,0}x+\tilde{\delta}_{0,0}^2)^i+u\underset{i=0}{\overset{p^s-1}{\sum}}(a_{0,i}^{1}+a_{1,i}^{1}x)(x^2+\tilde{\delta}_{0,0}x+\tilde{\delta}_{0,0}^2)^i\\&+\dots+u^{t-1}\underset{i=0}{\overset{p^s-1}{\sum}}(a_{0,i}^{t-1}+a_{1,i}^{t-1}x)(x^2+\tilde{\delta}_{0,0}x+\tilde{\delta}_{0,0}^2)^i,
 \end{align*}
        where $a_{j,i}^k\in \mathbb{F}_{p^m}$ for $0\leq j \leq 1,\,0\leq i\leq p^s-1,\,$ and $0\leq k\leq t-1.$
        Moreover, an arbitrary element $c(x)$ is a non-unit if and only if both $a_{0,0}^0$ and $a_{1,0}^0$ are zero. 
 Using this representation of an element of $R^{t,3}_\delta$, we can write the generator $u^{(t-1)-i}(x^2+\tilde{\delta}_{0,0}x+\tilde{\delta}_{0,0}^2)^{a_i}-u^{(t-1)-(i-1)}g(x)$ as 
\begin{align*}
u^{(t-1)-i}(x^2+\tilde{\delta}_{0,0}x&+\tilde{\delta}_{0,0}^2)^{a_{i}}+u^{(t-1)-(i-1)}(x^2+\tilde{\delta}_{0,0}x+\tilde{\delta}_{0,0}^2)^{t_{i-1,0}}g_{i-1,0}(x)+\dots \\&+ u^{(t-1)}(x^2+\tilde{\delta}_{0,0}x+\tilde{\delta}_{0,0}^2)^{t_{i-1,i-1}}g_{i-1,i-1}(x),
\end{align*}
where each $g_{i-1,j}(x)$ is either 0 or a unit in $\frac{\mathbb{F}_{p^m}[x]}{\langle (x^2+\tilde{\delta}_{0,0}x+\tilde{\delta}_{0,0}^2)^{p^s}\rangle}$ and $p^s-1\geq a_i>t_{i-1,0}>t_{i-1,1}>\dots>t_{i-1,i-1}\geq 0$. 

Using this, we can now give the following theorem describing the ideals of $\frac{R^t[x]}{\langle x^{2p^s}+\tilde{\delta}x^{p^s}+\tilde{\delta}^2\rangle}$. The proof of the theorem is similar to that of Corollary 3.6 of \cite{AkaKanSar}.
\begin{theorem}\label{ideals of Rt[x]x2psdeltaxpstildedelta2}
     Let $p^m\equiv 2 (\textnormal{mod }3)$ and let $\delta$ be a unit in $R^t$ which is a cube, that is, there exists $\tilde{\delta} = \tilde{\delta}_0+u\tilde{\delta}_1+\dots+u^{t-1}\tilde{\delta}_{t-1} \in R^t$ such that $\delta=\tilde{\delta}^3$. Let $\tilde{\delta}_0=\tilde{\delta}_{0, 0}^{p^s}$. Then the ideals of the ring
 $\frac{R^t[x]}{\langle x^{2p^s}+\tilde{\delta}x^{p^s}+\tilde{\delta}^2\rangle}$ have one of the following forms.
    \begin{itemize}
        \item [(i)] Trivial ideals $\langle 0\rangle,$ $\langle 1 \rangle.$
        \item [(ii)] Any generator of a non-trivial ideal contained in $\langle u\rangle $ has the form:
        \begin{align*}
        \theta_i(u, x^2+&\tilde{\delta}_{0,0}x+\tilde{\delta}_{0,0}^2)=u^{(t-1)-i}(x^2+\tilde{\delta}_{0,0}x+\tilde{\delta}_{0,0}^2)^{a_{i}}+u^{(t-1)-(i-1)}(x^2+\tilde{\delta}_{0,0}x+\\&\tilde{\delta}_{0,0}^2)^{t_{i-1,0}}g_{i-1,0}(x)+\dots + u^{(t-1)}(x^2+\tilde{\delta}_{0,0}x+\tilde{\delta}_{0,0}^2)^{t_{i-1,i-1}}g_{i-1,i-1}(x),
        \end{align*}
        for some $0\leq i \leq t-2$ and each $g_{i-1,j}(x)$ is either 0 or a unit in $\frac{\mathbb{F}_{p^m}[x]}{\langle (x^2+\tilde{\delta}_{0,0}x+\tilde{\delta}_{0,0}^2)^{p^s}\rangle}$ and $p^s-1\geq a_i>t_{i-1,0}>t_{i-1,1}>\dots>t_{i-1,i-1}\geq 0$ for $0\leq j\leq i-1$.\\
        In fact, any non-trivial ideal contained in $\langle u \rangle$ has the form:
        $$\langle \theta_{i_1}(u, x^2+\tilde{\delta}_{0,0}x+\tilde{\delta}_{0,0}^2),\dots, \theta_{i_d}(u, x^2+\tilde{\delta}_{0,0}x+\tilde{\delta}_{0,0}^2)\rangle,$$  where $0\leq i_1<i_2<\dots<i_d\leq t-2$, $0\leq a_{i_1}<a_{i_2}<\dots<a_{i_d}\leq p^s-1$, $0\leq t_{(i_1-1),(i_1-1)}<t_{(i_1-1),(i_1-2)}<\dots<t_{(i_1-1),0}<a_{i_1},\, 0\leq t_{(i_2-1),(i_2-1)}<t_{(i_2-1),(i_2-2)}<\dots<t_{(i_2-1),0}<a_{i_2},\, \dots,\,$ and $0\leq t_{(i_d-1),(i_d-1)}<t_{(i_d-1),(i_d-2)}<\dots<t_{(i_d-1),0}<a_{i_d}.$  
        \item[(iii)] Any non-trivial ideal not contained in $\langle u \rangle$ has the form:
        $$\langle (x^2+\tilde{\delta}_{0,0}x+\tilde{\delta}_{0,0}^2)^a+ uf(x)\rangle+I,$$
         where $f(x)\in \frac{R^t[x]}{\langle x^{2p^s}+\tilde{\delta}x^{p^s}+\tilde{\delta}^2\rangle}$, $I$ is an ideal described as in (ii), and $0\leq a_{i_d}<a\leq p^s-1$.
        \end{itemize}
\end{theorem}
As a special case, using Theorem \ref{ideals of Rt[x]x2psdeltaxpstildedelta2} and an argument similar to the one used to prove Corollary \ref{chainfornpscase}, we get the following corollary.
\begin{corollary}\label{k=1 case} Let $p^m\equiv 2 (\textnormal{mod }3).$ Then, 
    the ideals of $\frac{R^t[x]}{\langle x^{2p^s}+\tilde{\delta}x^{p^s}+\tilde{\delta}^2\rangle},$ where $\tilde{\delta}$ is a unit in $R^t$, form a chain if and only if $\tilde{\delta}_1\neq 0$. In fact,  the ideals of $\frac{R^t[x]}{\langle x^{2p^s}+\tilde{\delta}x^{p^s}+\tilde{\delta}^2\rangle}$ are
    $$\langle 0\rangle \subset \langle (x^{2}+\tilde{\delta}_{0,0}x+\tilde{\delta}_{0,0}^2)^{tp^s-1}\rangle\subset \langle (x^{2}+\tilde{\delta}_{0,0}x+\tilde{\delta}_{0,0}^2)^{tp^s-2}\rangle\subset \dots \subset \langle (x^{2}+\tilde{\delta}_{0,0}x+\tilde{\delta}_{0,0}^2)\rangle\subset \langle 1\rangle,$$
    where $\tilde{\delta}_0=\tilde{\delta}_{0, 0}^{p^s}$. Moreover, $x^{2}+\tilde{\delta}_{0,0}x+\tilde{\delta}_{0,0}^2 $
is nilpotent with index $tp^s$.
\end{corollary}
\par
Next, we assume that $\delta$ is not a cube. By Proposition \ref{delta power n iff delta_0 power n}, $\delta_{0}$ is not a cube. Also, using the notation in Section \ref{section2}, if $\delta_{0,0} \in \mathbb{F}_{p^m}$ is such that $\delta_{0,0}^{p^s}=\delta_0$ then $\delta_{0,0}$ is also not a cube. Thus, $x^3-\delta_{0,0} \in \mathbb{F}_{p^m}[x]$ is an irreducible polynomial. 
Hence, taking $n=3$ in Theorem \ref{nps ideals}, we get the generators and the structure of ideals of $R^{t,3}_{\delta}$, in this case.\\
Note that, by Corollary \ref{deg n-1 polynomial in nps case}, any non-zero polynomial of the form $ax^2+bx+c\in \mathbb{F}_{p^m}[x]$ is invertible  in $R^{t,3}_\delta$. Also, as in previous sections, if $k$ is the smallest integer such that $\delta_k \neq0$, where $1\leq k\leq t-1$, then $\langle (x^3-\delta_{0,0})^{p^s} \rangle=\langle u^k\rangle$, $x^3-\delta_{0,0}$ is nilpotent in $R^{t,3}_{\delta}$ with nilpotency index $\lceil \frac{t}{k}\rceil p^s$, and an arbitrary element $c(x)$ of $R^{t,3}_{\delta}$ can be uniquely written as
        \begin{align*}
            \underset{i=0}{\overset{p^s-1}{\sum}}(a_{0,i}^{0}+a_{1,i}^{0}x+a_{2,i}^0x^2)(x^3-\delta_{0,0})^i+\dots+u^{t-1}\underset{i=0}{\overset{p^s-1}{\sum}}(a_{0,i}^{t-1}+a_{1,i}^{t-1}x+a_{2,i}^{t-1}x^2)(x^3-\delta_{0,0})^i,
        \end{align*}
        where $a_{j,i}^k\in \mathbb{F}_{p^m}$ for $0\leq j \leq 2,\,0\leq i\leq p^s-1,\,$ and $0\leq k\leq t-1.$
        Further, an arbitrary element $c(x)$ of $R^{t,3}_{\delta}$, as given above, is a non-unit if and only if $a_{0,0}^0,\, a_{1,0}^0,\, a_{2,0}^0=0.$ 
        
Also, as in  Corollary \ref{chainfornpscase}, the ideals of $R^{t,3}_{\delta}$ form a chain if and only if $\delta_1\ne 0$. In fact, the ideals  of $R^{t,2}_{\delta}$, in this case, are
       $$\langle 0\rangle \subset \langle (x^{3}-\delta_{0,0})^{tp^s-1}\rangle\subset \langle (x^{3}-\delta_{0,0})^{tp^s-2}\rangle\subset \dots \subset \langle (x^{3}-\delta_{0,0})\rangle\subset \langle 1\rangle,$$
    where $\tilde{\delta}_0=\tilde{\delta}_{0, 0}^{p^s}$. Moreover, $x^{3}-\delta_{0,0}$ is nilpotent with index $tp^s$. 

In the next theorem, we summarize the description of ideals of the ring $R^{t,3}_{\delta}$ in all cases.
 \begin{theorem}\label{final theorem for section 5}
      Let $\delta=\delta_0+u\delta_1+\dots+u^{t-1}\delta_{t-1} \in R^t$ be a unit and let $\delta_{0, 0}^{p^s}=\delta_0$.
\begin{enumerate}
    \item[(a)] If $\delta$ is a cube in $R^t$, say $\delta=\tilde{\delta}^3$ and $p^m\equiv 1(\textnormal{mod }3 )$, then $x^{3p^s} - \delta =(x^{p^s} - \tilde{\delta})(x^{p^s} - b\tilde{\delta})(x^{p^s} - c\tilde{\delta})$, $R^{t, 3}_\delta\cong R^{t, 1}_{\tilde{\delta}}\times R^{t, 1}_{b\tilde{\delta}} \times R^{t, 1}_{c\tilde{\delta}}$, any ideal of $R^{t, 3}_\delta$ has the form $I_1 \times I_2 \times I_3$ where $I_1$, $I_2$, and $I_3$ are ideals of $R^{t,1}_{\tilde{\delta}}$, $R^{t,1}_{b\tilde{\delta}}$, and $R^{t,1}_{c\tilde{\delta}}$, respectively. The structure of $I_1$, $I_2$, and $I_3$ can be obtained by using Theorem \ref{nps ideals} with appropriate substitutions, as described above.
     \item [(b)]  If $\delta$ is a cube in $R^t$, say $\delta=\tilde{\delta}^3$ and $p^m\equiv 2 (\textnormal{mod }3)$,  then $R^{t,3}_{\delta}\cong R^{t,1}_{\tilde{\delta}}\times \frac{R^t[x]}{\langle x^{2p^s}+\tilde{\delta}x^{p^s}+\tilde{\delta}^2\rangle}$, and any ideal of $R^{t,3}_{\delta}$ has the form $I_1\times I_2,$ where $I_1$ is an ideal of $R^{t,1}_{\tilde{\delta}}$ the structure of which can be obtained using Theorem \ref{nps ideals} with appropriate substitutions, and $I_2$ is an ideal of $\frac{R^t[x]}{\langle x^{2p^s}+\tilde{\delta}x^{p^s}+\tilde{\delta}^2\rangle}$ the structure of which is described in Theorem \ref{ideals of Rt[x]x2psdeltaxpstildedelta2}.
     \item[(c)] If $\delta$ is not a cube, then the ideals of the ring $R^{t,3}_{\delta}$ are given by Theorem \ref{nps ideals} for $n=3$.
 \end{enumerate}
 \end{theorem}
Now, as a special case, we give all the types of ideals of $R^{3,3}_\delta = \frac{\frac{\mathbb{F}_{p^m}[u]}{\langle u^3\rangle}[x]}{\langle x^{3p^s}-\delta \rangle}$ and give their cardinalities, that is, we discuss the case when $t=3$. Let $\delta=\delta_0 + u\delta_1 + u^2\delta_2 \in R^3$ be a unit. If $\delta$ is a cube, say $\delta=\tilde{\delta}^3$, and $p^m\equiv 1 (\textnormal{mod } 3)$ then, as discussed at the beginning of this section, any ideal $I$ of $\frac{R^t[x]}{\langle x^{3p^s}-\delta\rangle}$ has the form $I_1 \times I_2 \times I_3$ where $I_1$ is an ideal of $\frac{R^t[x]}{\langle x^{p^s}-\tilde{\delta}\rangle}$, $I_2$ is an ideal of $\frac{R^t[x]}{\langle x^{p^s}-b\tilde{\delta}\rangle}$, and $I_3$ is an ideal of $\frac{R^t[x]}{\langle x^{p^s}-c\tilde{\delta}\rangle}$. We can, now, use Lemma \ref{Torsions} and Theorem \ref{Cardinality} to get the torsional degrees and cardinalities of $I_1$, $I_2$, and $I_3$. The cardinality of $I$ is, then, the product of cardinalities of $I_1$, $I_2$, and $I_3$. \\ 
Next, let $\delta$ be a cube, say $\delta=\tilde{\delta}^3$, and $p^m\equiv 2 (\textnormal{ mod } 3)$. Recall that, in this case, any ideal $I$ of $R^{3,3}_{\delta}$ has the form $I_1\times I_2$ where $I_1$ is an ideal of $\frac{R^3[x]}{\langle x^{p^s}-\tilde{\delta} \rangle}$ and $I_2$ is an ideal of $\frac{R^3[x]}{\langle x^{2p^s}+\tilde{\delta}x^{p^s}+\tilde{\delta}^2 \rangle}$. Now, the cardinality of $I_1$ can be obtained using Theorem \ref{Cardinality}. Thus, to obtain the cardinality of $I$, it is enough to obtain the cardinality of $I_2$. Let $\tilde{\delta}=\tilde{\delta}_0+u\tilde{\delta}_1+u^2\tilde{\delta}_2$. If $\tilde{\delta}_1\neq 0$ then by Corollary \ref{k=1 case}, the ideals of $\frac{R^3[x]}{\langle x^{2p^s}+\tilde{\delta}x^{p^s}+\tilde{\delta}^2 \rangle}$ where $\tilde{\delta}=\tilde{\delta}_0+u\tilde{\delta}_1+u^2\tilde{\delta}_2$ with $\tilde{\delta}_1\neq 0$. In fact, the ideals are
    $$\langle 0\rangle \subset \langle (x^{2}+\tilde{\delta}_{0,0}x+\tilde{\delta}_{0,0}^2)^{3p^s-1}\rangle\subset \langle (x^{2}+\tilde{\delta}_{0,0}x+\tilde{\delta}_{0,0}^2)^{3p^s-2}\rangle\subset \dots \subset \langle (x^{2}+\tilde{\delta}_{0,0}x+\tilde{\delta}_{0,0}^2)\rangle\subset \langle 1\rangle.$$
  Moreover, $x^{2}+\tilde{\delta}_{0,0}x+\tilde{\delta}_{0,0}^2 $
is nilpotent with index $3p^s$. For the cardinality of these ideals, we have the following theorem.
\begin{theorem}
    For $0 \leq i \leq 3p^s -1$, the cardinality of the ideal $\langle (x^{2}+\tilde{\delta}_{0,0}x+\tilde{\delta}_{0,0}^2)^{i}\rangle$ is $p^{m(3p^s-i)}$. 
\end{theorem}
\begin{proof} We have,

$(x^{2}+\tilde{\delta}_{0,0}x+\tilde{\delta}_{0,0}^2)^{i}\Bigg(\underset{j=0}{\overset{p^s-1}{\sum}}(a_{0,j}+b_{0,j}x)(x^{2}+\tilde{\delta}_{0,0}x+\tilde{\delta}_{0,0}^2)^{j}+u\underset{j=0}{\overset{p^s-1}{\sum}}(a_{1,j}+b_{1,j}x)(x^{2}+\tilde{\delta}_{0,0}x+\tilde{\delta}_{0,0}^2)^{j}\\+u^2\underset{j=0}{\overset{p^s-1}{\sum}}(a_{2,j}+b_{2,j}x)(x^{2}+\tilde{\delta}_{0,0}x+\tilde{\delta}_{0,0}^2)^{j}\Bigg)\\=\underset{j=i}{\overset{p^s-1}{\sum}}(a_{0,j}+b_{0,j}x)(x^{2}+\tilde{\delta}_{0,0}x+\tilde{\delta}_{0,0}^2)^{j}-u\big(\tilde{\delta}_1(x^{p^s}+2\tilde{\delta}_0)+u(\dots)\big)\underset{j=0}{\overset{i-1}{\sum}}(a_{0,j}+b_{0,j}x)(x^{2}+\tilde{\delta}_{0,0}x+\tilde{\delta}_{0,0}^2)^{j}+u\underset{j=i}{\overset{p^s-1}{\sum}}(a_{1,j}+b_{1,j}x)(x^{2}+\tilde{\delta}_{0,0}x+\tilde{\delta}_{0,0}^2)^{j}-u^2\big(\tilde{\delta}_1(x^{p^s}+2\tilde{\delta}_0)+u(\dots)\big)\underset{j=0}{\overset{i-1}{\sum}}(a_{1,j}+b_{1,j}x)(x^{2}+\tilde{\delta}_{0,0}x+\tilde{\delta}_{0,0}^2)^{j}+u^2\underset{j=i}{\overset{p^s-1}{\sum}}(a_{2,j}+b_{2,j}x)(x^{2}+\tilde{\delta}_{0,0}x+\tilde{\delta}_{0,0}^2)^{j}$.
   
   It follows that the cardinality of the ideal is $p^{m(3p^s-i)}.$\hfill{$\square$}
\end{proof}
Next, let  $\tilde{\delta}_1=0$ and $\tilde{\delta}_2= 0$ so that $\tilde{\delta}=\tilde{\delta}_0$. Also, using the notation from Section \ref{section2}, $\tilde{\delta}_{0,0}^{p^s}=\tilde{\delta}_0$. Hence, $x^{2p^s}+\tilde{\delta}x^{p^s}+\tilde{\delta}^2=x^{2p^s}+\tilde{\delta}_0x^{p^s}+\tilde{\delta}_0^2=(x^2+\tilde{\delta}_{0,0}x+\tilde{\delta}_{0,0}^2)^{p^s}$. Thus, 
the ideals of $\frac{R^3[x]}{\langle x^{2p^s}+\tilde{\delta}x^{p^s}+\tilde{\delta}^2 \rangle},$ are same as ideals of $\frac{R^3[x]}{\langle (x^2+\tilde{\delta}_{0,0}x+\tilde{\delta}_{0,0}^2)^{p^s}\rangle}$. The torsional degree of the ideals, in this case, can be determined using the approach in \cite{hesari2024torsion}, replacing $x-1$ with $x^2+\tilde{\delta}_{0,0}x+\tilde{\delta}_{0,0}^2$ and taking the coefficients in $\mathbb{F}_{p^m}+x\mathbb{F}_{p^m}$. The cardinality can, then, be computed using Theorem 4.5(iv) of \cite{AkaKanSar}.\\
Finally, let $\tilde{\delta}_1=0$ and $\tilde{\delta}_2\neq 0$. In this situation, we first list different types of ideals of the ring $\frac{R^3[x]}{\langle x^{2p^s}+\tilde{\delta}x^{p^s}+\tilde{\delta}^2 \rangle}$ in the next theorem.

\begin{theorem}\label{ideals of 3ps specific case1}
    The ideals of the ring $\frac{R^3[x]}{\langle x^{2p^s}+\tilde{\delta}x^{p^s}+\tilde{\delta}^2 \rangle}$ have one of the following eight types:
     \begin{enumerate}
     \item{\label{Type 1''}} $\langle 0 \rangle,\, \langle 1 \rangle.$
           \item{\label{Type 2''}} $\langle u^2 (x^{2}+\tilde{\delta}_{0,0}x+\tilde{\delta}_{0,0}^2)^a \rangle$ where $0\leq a \leq p^s-1.$
           \item{\label{Type 3''}} $\langle u(x^{2}+\tilde{\delta}_{0,0}x+\tilde{\delta}_{0,0}^2)^a+u^2(x^{2}+\tilde{\delta}_{0,0}x+\tilde{\delta}_{0,0}^2)^{t}h(x)\rangle$ where $0\leq t<L< a\leq p^s-1,$ $h(x)$ is either $0$ or a unit in $\frac{\mathbb{F}_{p^m}[x]}{\langle (x^{2}+\tilde{\delta}_{0,0}x+\tilde{\delta}_{0,0}^2)^{p^s}\rangle},$ and $L$ is the smallest non-negative integer such that $u^2(x^{2}+\tilde{\delta}_{0,0}x+\tilde{\delta}_{0,0}^2)^L\in \langle u(x^{2}+\tilde{\delta}_{0,0}x+\tilde{\delta}_{0,0}^2)^a+u^2(x^{2}+\tilde{\delta}_{0,0}x+\tilde{\delta}_{0,0}^2)^{t}h(x)\rangle.$
           
           \item{\label{Type 4''}}$\langle u(x^{2}+\tilde{\delta}_{0,0}x+\tilde{\delta}_{0,0}^2)^a+u^2(x^{2}+\tilde{\delta}_{0,0}x+\tilde{\delta}_{0,0}^2)^{t}h(x),\, u^2(x^{2}+\tilde{\delta}_{0,0}x+\tilde{\delta}_{0,0}^2)^b\rangle$ where $0\leq t<b<L< a\leq p^s-1$, $h(x)$ is either $0$ or a unit in $\frac{\mathbb{F}_{p^m}[x]}{\langle (x^{2}+\tilde{\delta}_{0,0}x+\tilde{\delta}_{0,0}^2)^{p^s}\rangle},$ and $L$ the smallest non-negative integer such that $u^2(x^{2}+\tilde{\delta}_{0,0}x+\tilde{\delta}_{0,0}^2)^{L}\in \langle u(x^{2}+\tilde{\delta}_{0,0}x+\tilde{\delta}_{0,0}^2)^a+u^2(x^{2}+\tilde{\delta}_{0,0}x+\tilde{\delta}_{0,0}^2)^{t}h(x)\rangle.$
           
           \item{\label{Type 5''}} $\langle (x^{2}+\tilde{\delta}_{0,0}x+\tilde{\delta}_{0,0}^2)^{a}+u(x^{2}+\tilde{\delta}_{0,0}x+\tilde{\delta}_{0,0}^2)^{t_0}h_0(x)+u^2(x^{2}+\tilde{\delta}_{0,0}x+\tilde{\delta}_{0,0}^2)^{t_1}h_1(x)\rangle$ where $0\leq t_1<t_0<L< a\leq p^s-1,\,0\leq t_1<M<L,\,$ $h_i(x)$ is either $0$ or a unit in $\frac{\mathbb{F}_{p^m}[x]}{\langle (x^{2}+\tilde{\delta}_{0,0}x+\tilde{\delta}_{0,0}^2)^{p^s}\rangle}$ for $i=0,\,1$ and $L,M$ are the smallest non-negative integers such that $u(x^{2}+\tilde{\delta}_{0,0}x+\tilde{\delta}_{0,0}^2)^L+u^2g(x),\, u^2(x^{2}+\tilde{\delta}_{0,0}x+\tilde{\delta}_{0,0}^2)^M\in \langle (x^{2}+\tilde{\delta}_{0,0}x+\tilde{\delta}_{0,0}^2)^{a}+u(x^{2}+\tilde{\delta}_{0,0}x+\tilde{\delta}_{0,0}^2)^{t_0}h_0(x)+u^2(x^{2}+\tilde{\delta}_{0,0}x+\tilde{\delta}_{0,0}^2)^{t_1}h_1(x)\rangle$ for some $g(x)\in \frac{\mathbb{F}_{p^m}[x]}{\langle (x^{2}+\tilde{\delta}_{0,0}x+\tilde{\delta}_{0,0}^2)^{p^s}\rangle}.$

           \item{\label{Type 6''}}$\langle (x^{2}+\tilde{\delta}_{0,0}x+\tilde{\delta}_{0,0}^2)^{a}+u(x^{2}+\tilde{\delta}_{0,0}x+\tilde{\delta}_{0,0}^2)^{t_0}h_0(x)+u^2(x^{2}+\tilde{\delta}_{0,0}x+\tilde{\delta}_{0,0}^2)^{t_1}h_1(x),\,u^2 (x^{2}+\tilde{\delta}_{0,0}x+\tilde{\delta}_{0,0}^2)^b \rangle\rangle$ for $0\leq t_1<t_0< a\leq p^s-1,\, 0\leq t_1<b <L<a,\,0\leq t_0<M<L$, $h_i(x)$ is either $0$ or a unit in $\frac{\mathbb{F}_{p^m}[x]}{\langle (x^{2}+\tilde{\delta}_{0,0}x+\tilde{\delta}_{0,0}^2)^{p^s}\rangle}$ for $i=0,\,1,$ and $L,M$ are the smallest non-negative integers such that $u^2(x^{2}+\tilde{\delta}_{0,0}x+\tilde{\delta}_{0,0}^2)^L,\, u(x^{2}+\tilde{\delta}_{0,0}x+\tilde{\delta}_{0,0}^2)^M+u^2g(x)\in \langle (x^{2}+\tilde{\delta}_{0,0}x+\tilde{\delta}_{0,0}^2)^{a}+u(x^{2}+\tilde{\delta}_{0,0}x+\tilde{\delta}_{0,0}^2)^{t_0}h_0(x)+u^2(x^{2}+\tilde{\delta}_{0,0}x+\tilde{\delta}_{0,0}^2)^{t_1}h_1(x) \rangle$ for some $g(x)\in \frac{\mathbb{F}_{p^m}[x]}{\langle (x^{2}+\tilde{\delta}_{0,0}x+\tilde{\delta}_{0,0}^2)^{p^s}\rangle}.$

           \item{\label{Type 7''}} $\langle (x^{2}+\tilde{\delta}_{0,0}x+\tilde{\delta}_{0,0}^2)^{a}+u(x^{2}+\tilde{\delta}_{0,0}x+\tilde{\delta}_{0,0}^2)^{t_0}h_0(x)+u^2(x^{2}+\tilde{\delta}_{0,0}x+\tilde{\delta}_{0,0}^2)^{t_1}h_1(x),\, u(x^{2}+\tilde{\delta}_{0,0}x+\tilde{\delta}_{0,0}^2)^b+u^2(x^{2}+\tilde{\delta}_{0,0}x+\tilde{\delta}_{0,0}^2)^{t_2}h_2(x)\rangle$ where $0\leq t_1<t_0<b<L< a\leq p^s-1\,, 0\leq t_2<b,\,0\leq t_i<M<b$ for $i=1,2$, $h_i(x)$ is either $0$ or a unit in $\frac{\mathbb{F}_{p^m}[x]}{\langle (x^{2}+\tilde{\delta}_{0,0}x+\tilde{\delta}_{0,0}^2)^{p^s}\rangle}$ for $0\leq i\leq 2,$ $L,M$ are the smallest non-negative integers such that $u(x^{2}+\tilde{\delta}_{0,0}x+\tilde{\delta}_{0,0}^2)^L+u^2g(x)\in \langle (x^{2}+\tilde{\delta}_{0,0}x+\tilde{\delta}_{0,0}^2)^{a}+u(x^{2}+\tilde{\delta}_{0,0}x+\tilde{\delta}_{0,0}^2)^{t_0}h_0(x)+u^2(x^{2}+\tilde{\delta}_{0,0}x+\tilde{\delta}_{0,0}^2)^{t_1}h_1(x)\rangle$ and $u^2(x^{2}+\tilde{\delta}_{0,0}x+\tilde{\delta}_{0,0}^2)^M\in \langle (x^{2}+\tilde{\delta}_{0,0}x+\tilde{\delta}_{0,0}^2)^{a}+u(x^{2}+\tilde{\delta}_{0,0}x+\tilde{\delta}_{0,0}^2)^{t_0}h_0(x)+u^2(x^{2}+\tilde{\delta}_{0,0}x+\tilde{\delta}_{0,0}^2)^{t_1}h_1(x),\, u(x^{2}+\tilde{\delta}_{0,0}x+\tilde{\delta}_{0,0}^2)^b+u^2(x^{2}+\tilde{\delta}_{0,0}x+\tilde{\delta}_{0,0}^2)^{t_2}h_2(x)\rangle$ for some $g(x)\in \frac{\mathbb{F}_{p^m}[x]}{\langle (x^{2}+\tilde{\delta}_{0,0}x+\tilde{\delta}_{0,0}^2)^{p^s}\rangle}.$

           \item{\label{Type 8''}}$\langle (x^{2}+\tilde{\delta}_{0,0}x+\tilde{\delta}_{0,0}^2)^{a}+u(x^{2}+\tilde{\delta}_{0,0}x+\tilde{\delta}_{0,0}^2)^{t_0}h_0(x)+u^2(x^{2}+\tilde{\delta}_{0,0}x+\tilde{\delta}_{0,0}^2)^{t_1}h_1(x),\, u(x^{2}+\tilde{\delta}_{0,0}x+\tilde{\delta}_{0,0}^2)^b+u^2(x^{2}+\tilde{\delta}_{0,0}x+\tilde{\delta}_{0,0}^2)^{t_2}h_2(x),\, u^2(x^{2}+\tilde{\delta}_{0,0}x+\tilde{\delta}_{0,0}^2)^c\rangle$ where $0\leq t_1<t_0< a\leq p^s-1,\, 0\leq t_1,t_2<c<b<a,\,0\leq t_0<b$,  $h_i(x)$ is either $0$ or a unit in $\frac{\mathbb{F}_{p^m}[x]}{\langle (x^{2}+\tilde{\delta}_{0,0}x+\tilde{\delta}_{0,0}^2)^{p^s}\rangle}$ for $0\leq i\leq 2$, $c<M$, the smallest non-negative integer such that $u^2(x^{2}+\tilde{\delta}_{0,0}x+\tilde{\delta}_{0,0}^2)^M\in \langle (x^{2}+\tilde{\delta}_{0,0}x+\tilde{\delta}_{0,0}^2)^{a}+u(x^{2}+\tilde{\delta}_{0,0}x+\tilde{\delta}_{0,0}^2)^{t_0}h_0(x)+u^2(x^{2}+\tilde{\delta}_{0,0}x+\tilde{\delta}_{0,0}^2)^{t_1}h_1(x),\, u(x^{2}+\tilde{\delta}_{0,0}x+\tilde{\delta}_{0,0}^2)^b+u^2(x^{2}+\tilde{\delta}_{0,0}x+\tilde{\delta}_{0,0}^2)^{t_2}h_2(x)\rangle$, and $b<L$, the smallest non-negative integer such that $u(x^{2}+\tilde{\delta}_{0,0}x+\tilde{\delta}_{0,0}^2)^L+u^2g(x)\in \langle (x^{2}+\tilde{\delta}_{0,0}x+\tilde{\delta}_{0,0}^2)^{a}+u(x^{2}+\tilde{\delta}_{0,0}x+\tilde{\delta}_{0,0}^2)^{t_0}h_0(x)+u^2(x^{2}+\tilde{\delta}_{0,0}x+\tilde{\delta}_{0,0}^2)^{t_1}h_1(x)\rangle$ for some $g(x)\in \frac{\mathbb{F}_{p^m}[x]}{\langle (x^{2}+\tilde{\delta}_{0,0}x+\tilde{\delta}_{0,0}^2)^{p^s}\rangle}.$
       \end{enumerate}
\end{theorem}
As in Section \ref{section3}, we can give the torsional degree as well as cardinalities of the ideals in Theorem \ref{ideals of 3ps specific case1}. The results are similar to Lemma \ref{Torsions} and Theorem \ref{Cardinality}. The critical piece for obtaining the torsional degree and cardinality is once again obtaining the parameters used in the description of ideals in Theorem \ref{ideals of 3ps specific case1}. For the sake of completeness, we state here the results giving the parameters. The proofs of these results and computations are similar to those of Lemmas \ref{first parameter computation}, Lemma \ref{second parameter computation}, and Lemma \ref{third parameter computation}, respectively. The main differences include using $c_i(x)=\underset{j=0}{\overset{p^s-1}{\sum}}(a_{i,j}+b_{i,j}x)(x^{2}+\tilde{\delta}_{0,0}x+\tilde{\delta}_{0,0}^2)^j$, where $a_{i,j},b_{i,j}\in \mathbb{F}_{p^m}$ for $i=0,1,\,2$ and $0\leq j\leq p^s-1$, writing $c_0(x)=c_{0,0}(x)+(x^{2}+\tilde{\delta}_{0,0}x+\tilde{\delta}_{0,0}^2)^{p^s-a}c_{0,1}(x)$ where $c_{0,0}(x)=\underset{j=0}{\overset{p^s-1-a}{\sum}}(a_{0,j}+b_{0,j}x)(x^{2}+\tilde{\delta}_{0,0}x+\tilde{\delta}_{0,0}^2)^j$ and $c_{0,1}(x)=\underset{j=0}{\overset{a-1}{\sum}}(a_{0,{j+p^s-a}}+b_{0,{j+p^s-a}}x)(x^{2}+\tilde{\delta}_{0,0}x+\tilde{\delta}_{0,0}^2)^j,$ and replacing $(x-\delta_{0,0})$ with $(x^{2}+\tilde{\delta}_{0,0}x+\tilde{\delta}_{0,0}^2).$ Further, the cardinality is given by Theorem 4.6(iv) in \cite{AkaKanSar}, that is, for an ideal $C$ in Theorem \ref{ideals of 3ps specific case1}, the cardinality of $C=(p^{2m})^{3p^s-\underset{i=0}{\overset{2}{\sum}}T_i}$.
\begin{lemma}
        Let $L$ be the smallest non-negative integer such that $u^2(x^{2}+\tilde{\delta}_{0,0}x+\tilde{\delta}_{0,0}^2)^L\in \langle u(x^{2}+\tilde{\delta}_{0,0}x+\tilde{\delta}_{0,0}^2)^a+u^2(x^{2}+\tilde{\delta}_{0,0}x+\tilde{\delta}_{0,0}^2)^{t}h(x)\rangle,$ where $h(x)$, if non-zero, is a unit in $\frac{\mathbb{F}_{p^m}[x]}{\langle (x^{2}+\tilde{\delta}_{0,0}x+\tilde{\delta}_{0,0}^2)^{p^s}\rangle}$. Then
        $$L= \begin{cases} 
             a & \textnormal{ if } h(x)=0, \\
             \Min\{a,\,p^s-a+t\}  & \textnormal{ if } h(x) \ne 0. \\
   \end{cases}$$
\end{lemma}
\begin{lemma}
        Let $L$ be the smallest non-negative integer such that $u^2(x^{2}+\tilde{\delta}_{0,0}x+\tilde{\delta}_{0,0}^2)^L\in \langle (x^{2}+\tilde{\delta}_{0,0}x+\tilde{\delta}_{0,0}^2)^{a}+u(x^{2}+\tilde{\delta}_{0,0}x+\tilde{\delta}_{0,0}^2)^{t_0}h_0(x)+u^2(x^{2}+\tilde{\delta}_{0,0}x+\tilde{\delta}_{0,0}^2)^{t_1}h_1(x)\rangle$, where for $i=0$ and $1$, $h_i(x)$, if non-zero, is a unit in $\frac{\mathbb{F}_{p^m}[x]}{\langle (x^{2}+\tilde{\delta}_{0,0}x+\tilde{\delta}_{0,0}^2)^{p^s}\rangle}$.Then
        $$L= \begin{cases} 
             0 & \textnormal{ if } h_0(x)=0, \\
             0&\textnormal{ if } h_0(x)\neq 0 \textnormal{ and } a<p^s-a+t_0,\\
             \Min\{a,\,\beta\}  & \textnormal{ if } h_0(x) \ne 0 \textnormal{ and } a\geq p^s-a+t_0, \\
   \end{cases}$$
   where $\beta:=\max\{k:(x^{2}+\tilde{\delta}_{0,0}x+\tilde{\delta}_{0,0}^2)^k\mid ((x^{2}+\tilde{\delta}_{0,0}x+\tilde{\delta}_{0,0}^2)^{t_0}h_0(x) -h_0(x)^{-1}(x^{2}+\tilde{\delta}_{0,0}x+\tilde{\delta}_{0,0}^2)^{2a-p^s-t_0}-h_1(x)h_0(x)^{-1}(x^{2}+\tilde{\delta}_{0,0}x+\tilde{\delta}_{0,0}^2)^{a+t_1-t_0}h_1(x)h_0(x)^{-1})\}.$
\end{lemma}
\begin{lemma}
        Let $L$ be the smallest non-negative integer such that $u^2(x^{2}+\tilde{\delta}_{0,0}x+\tilde{\delta}_{0,0}^2)^L\in \langle (x^{2}+\tilde{\delta}_{0,0}x+\tilde{\delta}_{0,0}^2)^{a}+u(x^{2}+\tilde{\delta}_{0,0}x+\tilde{\delta}_{0,0}^2)^{t_0}h_0(x)+u^2(x^{2}+\tilde{\delta}_{0,0}x+\tilde{\delta}_{0,0}^2)^{t_1}h_1(x),\, u(x^{2}+\tilde{\delta}_{0,0}x+\tilde{\delta}_{0,0}^2)^b+u^2(x^{2}+\tilde{\delta}_{0,0}x+\tilde{\delta}_{0,0}^2)^{t_2}h_2(x)\rangle$  where for $i=0,1$ and $2$, $h_i(x)$, if non-zero, is a unit in $\frac{\mathbb{F}_{p^m}[x]}{\langle (x^{2}+\tilde{\delta}_{0,0}x+\tilde{\delta}_{0,0}^2)^{p^s}\rangle}$. Then
        $$L= \begin{cases} 
             0 & \textnormal{ if } h_0(x)=0, \\
             0  & \textnormal{ if } h_0(x) \ne 0 \textnormal{ and } a<p^s-a+t_0, \\
              \Min\{b,\,\beta_1,\,\beta_2\}  & \textnormal{ if } h_0(x) \ne 0 \textnormal{ and } b\geq p^s-a+t_0, \\
             
   \end{cases}$$
   where $\beta_1:=\max\{k:(x^{2}+\tilde{\delta}_{0,0}x+\tilde{\delta}_{0,0}^2)^k\mid ((x^{2}+\tilde{\delta}_{0,0}x+\tilde{\delta}_{0,0}^2)^{t_0}h_0(x) -(x^{2}+\tilde{\delta}_{0,0}x+\tilde{\delta}_{0,0}^2)^{2a-p^s-t_0}h_0(x)^{-1}-(x^{2}+\tilde{\delta}_{0,0}x+\tilde{\delta}_{0,0}^2)^{a+t_1-t_0}h_1(x)h_0(x)^{-1})\}$ and $\beta_2:=\max\{k: (x^{2}+\tilde{\delta}_{0,0}x+\tilde{\delta}_{0,0}^2)^k\mid ((x^{2}+\tilde{\delta}_{0,0}x+\tilde{\delta}_{0,0}^2)^{t_2}h_2(x)-(x^{2}+\tilde{\delta}_{0,0}x+\tilde{\delta}_{0,0}^2)^{a+b-p^s-t_0}h_0(x)^{-1}-(x^{2}+\tilde{\delta}_{0,0}x+\tilde{\delta}_{0,0}^2)^{b+t_1-t_0}h_1(x)h_0(x)^{-1})\}.$
\end{lemma}
Now, let $\delta$ be not a cube in $R^3$. As in Theorem \ref{ideals whose parameters are to be computed}, the ring $R^{3,3}_{\delta}$, in this case, has eight different types of ideals. Replacing $x^2-\delta_{0,0}$ in the ideal types given in Theorem \ref{ideals whose parameters are to be computed} with $x^3-\delta_{0,0}$ gives the ideal types in this case. The parameters which help us in computing the torsional degree and cardinality of the ideals, in this case, can be computed in a similar manner as Lemma \ref{first parameter computation}, Lemma \ref{second parameter computation}, and Lemma \ref{third parameter computation}. The main differences include writing $c_i(x)=\underset{j=0}{\overset{p^s-1}{\sum}}(a_{i,j}+b_{i,j}x+c_{i,j}x^2)(x^{3}-\delta_{0,0})^j$, where $a_{i,j},b_{i,j},c_{i,j}\in \mathbb{F}_{p^m}$ for $0\leq i\leq 2$ and $0\leq j \leq p^s-1$, writing $c_0(x)=c_{0,0}(x)+(x^{3}-\delta_{0,0})^{p^s-a}c_{0,1}(x)$ where $c_{0,0}(x)=\underset{j=0}{\overset{p^s-1-a}{\sum}}(a_{0,j}+b_{0,j}x+c_{0,j}x^2)(x^{3}-\delta_{0,0})^j$ and $c_{0,1}(x)=\underset{j=0}{\overset{a-1}{\sum}}(a_{0,{j+p^s-a}}+b_{0,{j+p^s-a}}x+c_{0,j+p^s-a}x^2)(x^{3}-\delta_{0,0})^j,$ and replacing $(x-\delta_{0,0})$ with $(x^{3}-\delta_{0,0})$. The torsional degree ($T_i$ for $0\leq i\leq 2$) is, then, computed as in Lemma \ref{Torsions} and the cardinalities of the ideal are, then, given by Theorem 4.6(iv) in \cite{AkaKanSar}, that is, for an ideal $C$ of $R^{3,3}_{\delta}$, the cardinality of $C=(p^{3m})^{3p^s-\underset{i=0}{\overset{2}{\sum}}T_i}$.
\section{Conclusion}\label{section6}
In this paper, we studied the structure of ideals in polynomial quotient rings over the finite chain ring $R^t=\frac{\mathbb{F}_{p^m}[u]}{\langle u^t\rangle}$. By determining explicit generators, we obtained a complete description of the ideals of
$
\frac{R^t[x]}{\langle x^{np^s}-\delta\rangle},
$
where $\delta $ is a unit in $R^t$ and $(n,p)=1$. These results were then applied to several special cases corresponding to constacyclic codes of lengths $p^s$, $2p^s$, and $3p^s$. In particular, for the case $t=3$, we derived explicit parameters that help us in the computation of the cardinalities of the corresponding constacyclic codes. The results obtained in this work not only generalize several known classifications of constacyclic codes over finite fields, but also provide a unified framework for analyzing constacyclic codes over $R^t$. These structural insights and cardinality formulas may serve as a foundation for further investigations into dual codes, distance for constacyclic codes over finite chain rings.\\

\textbf{Conflict of Interest.} All authors declare that they have no conflict of interest.
\section*{Acknowledgements}
The first author would like to acknowledge PMRF (PMRF Id: 1403187) for its financial support.
\bibliographystyle{abbrv}
\bibliography{sample}

\end{document}